\documentclass{llncs}

\usepackage{booktabs}   

\usepackage{comment}
\usepackage{arydshln}
\usepackage{MnSymbol}
\usepackage{amsfonts}

\usepackage{tikz}
\usetikzlibrary{arrows,shapes,patterns,calc}

\usepackage{wrapfig}
\usepackage[caption=false,font=footnotesize]{subfig}

\usepackage{rotating}

\newcommand{\dlangle}{\langle \! \langle}
\newcommand{\drangle}{\rangle \! \rangle}
\newcommand{\G}{\mathcal G}

\newcommand\altsim[1]{\Rrightarrow_{#1}}
\newcommand{\altbisim}[1]{\Lleftarrow\!\!\!\Rrightarrow_{#1}}
\newcommand{\br}[1]{\overline{#1}}

\def\restr#1{\,\mbox{\rule[-4pt]{0.5pt}{13pt}}_{#1}}

\tikzstyle{every node} =
[draw = none, fill = white, thin]
\tikzstyle{every edge} +=
[black, thick]

\tikzstyle{noall} =
[draw = none, fill = none]
\tikzstyle{nodraw} =
[draw = none, fill = white]
\tikzstyle{nofill} =
[draw = black, fill = none]

\tikzstyle{cnode} =
[circle, draw = black]
\tikzstyle{snode} =
[regular polygon, regular polygon sides = 4, draw = black]
\tikzstyle{lnode} =
[diamond, draw = black]

\tikzstyle{initstate}=[circle,draw,trans, minimum size=8mm, fill=yellow]
\tikzstyle{state}=[circle,draw,trans, minimum size=8mm]
\tikzstyle{trans}=[font=\footnotesize]

\newcommand\cata[1]{\textcolor{black}{#1}}

\newcommand\new[1]{\textcolor{black}{#1}}

\def\duplic{{\sc{Duplicator}}}
\def\spoiler{{\sc{Spoiler}}}

\def\pspoil{{\sc{P-Spoil}}}
\def\ispoil{{\sc{I-Spoil}}}
\def\pdupl{{\sc{P-Dupl}}}
\def\idupl{{\sc{I-Dupl}}}

\def\B{{\mathcal{B}}}

\definecolor{myred}{cmyk}{0, 0.7808, 0.4429, 0.1412}
\definecolor{myblu}{rgb}{0.5, 0.7, 0.9}

\newcommand\gtm{\G_{M}}
\newcommand\gs{\G_{S}} 

\title{A Hennessy-Milner Theorem for ATL\\ with Imperfect Information}

\author{Francesco Belardinelli$^{1,2}$ \and Catalin Dima$^{3}$ \and Vadim Malvone$^{2}$ \and Ferucio Tiplea$^{4}$}

\institute{Imperial College London, UK \and Universit\'{e} d'\'{E}vry, France \and Universit\'{e} Paris-Est Cr\'{e}teil, France \and Universitatea Al. I. Cuza, Ia\c si, Romania}

\begin{document}

\maketitle

\begin{abstract}
	We show that a history-based variant of alternating bisimulation with
	imperfect information allows it to be related to a variant of Alternating-time Temporal Logic (ATL) with
	imperfect information by a full Hennessy-Milner theorem.  The variant
	of ATL we consider has a {\em common knowledge} semantics, which requires
	that the uniform strategy available for a coalition to accomplish some
	goal must be common knowledge inside the coalition, while other
	semantic variants of ATL with imperfect information do not accomodate
	a Hennessy-Milner theorem.  We also show that the existence of
	a history-based alternating bisimulation between two finite Concurrent
	Game Structures with imperfect information (iCGS) is undecidable.
\end{abstract}

\section{Introduction} \label{introduction}




Alternating-time Temporal Logic (ATL) \cite{AlurHenzingerKupferman02}
is a powerful logic for specifying strategic abilities of individual
agents and coalitions in multi-agent game structures.  Crucially, ATL
has been extended to games of imperfect information \cite{Jamroga+04b}
with various flavors related to the agents' knowledge of the existence
of strategies for accomplishing the coalition's
goals \cite{Agotnes+15,BullingDixJamroga10a,BullingJamroga14}.  In
this contribution, we focus on the {\em common knowledge} ($ck$)
interpretation of ATL under imperfect information, which was first
put forward in
\cite{Jamroga+04b}, along with its {\em objective} and {\em subjective}
interpretations. However, differently from the latter, to the best of
our knowledge, the $ck$ interpretation has nowhere else been
considered in the literature.  Nonetheless, the $ck$ interpretation
allows us to prove a Hennessy-Milner theorem for ATL under imperfect
information for the memoryful notion of bisimulation we introduce in
this paper. This result is in marked contrast with the situation for
the other interpretations, which do not enjoy the Hennessy-Milner
property \cite{HennessyM80}.

The literature on bisimulations for modal logics is extensive, an
in-depth survey of model equivalences for various temporal logics
appears in \cite{Goltz92equivalences}. The landscape for logics of
strategic abilities, including ATL, is comparatively more sparse.
%
A proof of the Hennessy-Milner property for ATL$^*$ with perfect
information was already given in the paper introducing alternating
bisimulations \cite{AlurHKV98}.  Since then, there have been numerous
attempts to extend bisimulations to more expressive languages
(including Strategy Logic recently \cite{BelardinelliDM18}), as well
as to contexts of imperfect information
\cite{Agotnes07irrevocable,Dastani10programs-aamas,BelardinelliCDJJ17}.
In \cite{Dastani10programs-aamas,Melissen13phd} non-local model
equivalences for ATL with imperfect information have been put
forward. However, these works do not deal with the imperfect
information/perfect recall setting here considered, nor do they
provide a local account of bisimulations.  Further,
in \cite{BelardinelliCDJJ17} the authors consider a {\em memoryless}
notion of bisimulation for ATL, under imperfect
information. Unfortunately, their definition does not allow for the
Hennessy-Milner property.  
We also note the results from \cite{DimaMP18} which show that ATL
with imperfect information is incompatible in expressive power when compared with 
the modal-epistemic $\mu$-calculus, contrary to what is known for the perfect information case.
The present contribution extends the notion
of alternating bisimulation to the setting of imperfect information
and perfect recall so that it satisfies the Hennessy-Milner property:
two game structures are bisimilar iff they satisfy the same formulas
in ATL.

The classic proof for Hennessy-Milner type properties
typically uses bisimulation games played between \spoiler{}
and \duplic{}. These bisimulation games
are turn-based, perfect information, safety games (in regards
of \duplic's objective) played on a tree whose nodes are labeled with
pairs of states (or histories, in case of a memoryful semantics) of
the two game structures checked for bisimulation.  Hence,
such games are determined, 
and determinacy plays a
crucial role since, when there is no bisimulation between the two
structures, the bisimulation game cannot be won by \duplic{}, and
hence \spoiler{} has a winning strategy, which is then used for
exhibiting a formula that is satisfied in one structure but not in the
other.

%


The extension of this proof technique to ATL with imperfect
information has to cope with the fact that any notion of bisimulation has
to account for the fact that
coalitions have to choose action profiles in indistinguishable states
in a ``uniform'' way: agents that do not distinguish between two
states must choose the same actions in both.  Uniformity entails a
slightly more involved notion of bisimulation which utilizes {\em
	strategy simulators} \cite{BelardinelliCDJJ17}.
Then, any bisimulation game has to encode these strategy simulators,
in the sense that \duplic{} is given the choice of a uniform strategy
in some common-knowledge neighbourhood in one of the game structures
and the \spoiler{} has to reply with a uniform strategy in the
related common-knowledge neighbourhood of the other game structure.

The problem raised by this generalization is that positions in a
bisimulation game are normally labeled with histories, not
common-knowledge neighbourhoods, as bisimulations relate the former, not the latter.
So, we need both a \spoiler{} and a \duplic{} who have imperfect
information at each position of the bisimulation game.  On the other
hand, as it is the case with bisimulations for the perfect information
case, for each choice of strategies in the two structures, the outcomes of
one strategy have to be related with the outcomes of the other strategy.
But this requires both \spoiler{} and \duplic{} to
be \emph{perfectly-informed}!

The solution we propose is a 4-player bisimulation game played between  
the  \spoiler{} coalition $ \{$\ispoil, \pspoil $\}$ 
and the \duplic{} coalition $\{$\idupl, \pdupl $\}$, where both {\sc
	I}-players have imperfect information, while both {\sc P}-players have
perfect information.  We show that such a game can be won by
the \duplic{} coalition if and only if there exists a bisimulation
between the two game structures.

Further, we provide a Gale-Stewart type determinacy
theorem \cite{GaleS53} for the bisimulation game, showing that exactly
one of the two coalitions has a winning joint strategy.
The key point is that, when \duplic{} does not have a winning
strategy, the strategic choices for \ispoil{} can be defined in a
uniform way that is only based on \ispoil's observations.  To the best of our
knowledge, this is the first example of a class of multi-player,
imperfect information, zero-sum (reachability) games played over
infinite trees that are determined.
Note that, for technical reasons, our Hennessy-Milner theorem only holds for ATL with the ``yesterday'' modality~$\mathbf{Y}$. 

Moreover, we analyse the problem of checking the existence of a
bisimulation between two given game structures.  We show that this
problem is undecidable in general by building on the undecidability of
the model-checking problem for ATL with imperfect information and
perfect recall.  More specifically, given a Turing machine $M$, we
build a game structure in which a two-agent coalition has a strategy
for avoiding an error state if and only if $M$ halts when starting
with an empty tape. We then build a second, unrelated, simple game
structure in which the same coalition always has an avoiding
strategy. Finally, we prove that the two structures are bisimilar if
and only if $M$ halts.

{\bf Scheme of the paper.}  In Sec.~\ref{preliminaries} we recall the
syntax and semantics of ATL according to various flavors of imperfect
information (and perfect recall). Sec.~\ref{altbisim} extends the
bisimulation in \cite{BelardinelliCDJJ17} to the case of perfect
recall, and shows that bisimilar game structures satisfy the same
formulas in ATL.  Then, in Sec.~\ref{bisimgames} we introduce our
variant of bisimulation games, for which we prove that the \duplic{}
coalition has a winning strategy if and only if there exists a
bisimulation between the two given game structures.  In Sec.~\ref{HMproperty} we
prove the Gale-Stewart determinacy theorem for our bisimulation games,
which allows us to prove the
Hennessy-Milner theorem.  Finally, in Sec.~\ref{bisim_undec} we show
that
checking the existence of a bisimulation between two
given game structures is undecidable in general.

\section{ATL with Imperfect Information} \label{preliminaries}


In this section we present the syntax and semantics of the
Alternating-time Temporal Logic ATL$^*$
\cite{AlurHenzingerKupferman02}.  In the rest of the paper we assume a
set $AP$ of atomic propositions (or atoms) and a set $Ag$ of agents.
%
%
\begin{definition}[ATL$^*$]
	History formulas $\phi$ and path formulas $\psi$ in ATL$^*$ are defined
	by the following BNF, where $p \in AP$ and $A \subseteq Ag$:
	\begin{eqnarray*}
		\phi  & ::= &
		p \mid \neg \phi \mid \phi \to \phi \mid
		\dlangle A \drangle \psi\\
		\psi  & ::= &
		\phi \mid \neg \psi \mid \psi \to \psi \mid
		X \psi \mid \mathbf{Y} \psi \mid \psi U \psi
	\end{eqnarray*}
	
	The formulas in ATL$^*$ are all and only the history formulas.
\end{definition}

The ATL$^*$ operator $\dlangle A \drangle$ intuitively means that `the
agents in coalition $A$ have a (collective) strategy to achieve \ldots',
where the goals are LTL formulas built by using operators `next' $X$ and 
`until' $U$.
We
define $A$-formulas as the formulas in ATL$^*$ in which $A$ is the only
coalition appearing in ATL$^*$ modalities.

Notice that we talk about {\em history} formulas, rather than {\em
	state} formulas as customary, as such formulas will be interpreted
on histories rather than states as per perfect recall.

We provide ATL$^*$ with both the objective and subjective variants
\cite{Jamroga+04b} of the history-based semantics with imperfect
information and perfect recall, as well as a novel interpretation
based on {\em common knowledge} \cite{Fagin+95b}.



\begin{definition}[iCGS] \label{iCGS}
	Given sets $AP$ of atoms and $Ag$ of agents, a {\em concurrent game
		structure with imperfect information}, or iCGS, is a tuple $\G =
	\langle Ag, S, s_0, Act, \allowbreak \{\sim_i \}_{i \in Ag},
	d,\to, \pi \rangle$ where
	\begin{itemize}
		\item  $S$ is a non-empty set of \emph{states} and $s_0 \in S$ is the \emph{initial state} of $\G$.
		\item $Act$ is a finite non-empty set of {\em actions}. A tuple $\vec a  = ( a_i)_{i \in Ag} \in Act^{Ag}$
		is called a \emph{joint action}.
		
		\item For every agent $i \in Ag$, $\sim_i$ is an equivalence relation on $S$,
		called the \emph{indistinguishability relation} for agent $i$.
		
		\item $d: Ag \times S \to (2^{Act} \setminus \{ \emptyset \})$ is the \emph{protocol function},
		satisfying the property that, for all states $s, s'\in S$ and any agent $i$, $s \sim_i s'$ implies $d(i, s) = d(i,
		s')$. That is, the same (non-empty) set of  actions is available to agent $i$ in indistinguishable states.
		\item $\to \subseteq S \times Act^{Ag} \times S$ is the
		{\em transition relation} such that, for every state $s \in S$ and
		joint action $\vec{a} \in Act^{Ag}$, $(s,\vec{a},s') \in \to$
		for some state $s' \in S$ iff $a_i \in d(i, s)$ for every agent $i \in
		Ag$.  We normally write $s \xrightarrow{\vec{a}} r$ for
		$(s,\vec{a},r) \in \to$.
		\item
		$\pi: S \to 2^{AP}$ is the {\em state-labeling function}.
	\end{itemize}
\end{definition}

\paragraph{Runs.}
Given an iCGS $\G$, a {\em run} is a finite or infinite
sequence $\rho = s_0 \vec{a}_0 s_1 \ldots$ in $((S \cdot
Act^{Ag})^*\cdot S) \cup (S \cdot Act^{Ag})^\omega$ such that for
every $j\geq 0$, $s_j \xrightarrow{\vec{a}_j} s_{j+1}$. 
Given a run
$\rho = s_0 \vec{a}_0 s_1 \ldots$ and $j \geq 0$, $\rho[j]$
denotes the $j+1$-th state $s_j$ in the sequence and $\rho[j,k]$ denotes the sequence of states from the $j+1$-th state to the $k+1$-th state; while $\rho_{\geq
	j}$ (or $\rho[\geq j]$) denotes run $s_j \vec{a}_j s_{j+1} \ldots$ starting from
$\rho[j]$, and $\rho_{\leq j}$ (or $\rho[\leq j]$) denotes run $s_0 \vec{a}_0 s_{1} \ldots \vec{a}_{j-1} s_{j} $. 
Further, with $act_i (h,m)$ we denote the $m$-th action of agent $i$ in history $h$.

We call finite runs {\em histories}, denote them as $h \in H$, their
length as $|h| \in \mathbb{N}$, and their last element $h_{|h|-1}$ as
$last(h)$; whereas infinite runs are called {\em paths} and denoted as
$\lambda, \lambda' \in P$.  We denote the set of all histories
(resp.~paths) in an iCGS $\G$ as $Hist(\G)$ (resp.~$Path(\G)$).
Notice that states are instances of histories of length
1. Accordingly, several notions defined below for histories can also
by applied to states.
%
Finally, we write $h \preceq \rho$ to say that $h$ is the prefix of $\rho$, that is $h = \rho[\leq |h|]$.

For a coalition $A \subseteq Ag$ of agents, a {\em joint $A$-action}
denotes a tuple $\vec a_A = (a_i)_{i \in A} \in Act^A$ of actions, one
for each agent in $A$.  For coalitions $A\subseteq B \subseteq Ag$ of
agents, a joint $A$-action $\vec a_A$ \emph{is extended} by a joint
$B$-action $\vec b_B$, denoted $\vec a_A \sqsubseteq \vec b_B$, if for
every $i \in A$, $a_i = b_i$.  Also, a joint $A$-action $\vec
a_A$ is \emph{enabled} at state $s\in S$ if for every agent $i \in A$,
$a_i \in d(i,s)$.

\paragraph{Epistemic neighbourhoods.}  We extend the indistinguishability
relations $\sim_i$, for $i\in Ag$, to histories in a synchronous,
point-wise manner: $h \sim_i h'$ iff $|h| = |h'|$ and for all $m \leq
|h|$, $h_m \sim_i h'_m$ and $act_i(h,m) = act_i(h',m)$.

Given a coalition $A \subseteq Ag$ of agents, the \emph{collective
	knowledge relation} $\sim^E_A$ is defined as $\bigcup_{i \in A}
\sim_i$, while the \emph{common knowledge relation} $\sim^C_A$ is the
transitive closure $(\bigcup_{i \in A} \sim_i)^+$ of $\sim^E_A$.
Then, $C^{\G}_A(h) = \{ h' \in H \mid h' \sim_A^C h \}$ is the {\em
	common knowledge neighbourhood} (CKN) of history $h$ for coalition
$A$ in the iCGS $\G$.
We will omit the
superscript $\G$ whenever it is clear from the context.

\paragraph{Uniform strategies.}
We introduce a notion of strategy for the interpretation of $\dlangle
A \drangle$ modalities.
\begin{definition}[Strategy] \label{unstrategy}
	A \emph{(uniform, memoryfull)} {\em strategy} for an agent $i \in Ag$ is a
	function $\sigma : H \to Act$ that is compatible with $d$ and
	$\sim_i$, that is, for all histories $h, h' \in H$,
	$\sigma(h) \in d(i, last(h))$
	and 
	$h \sim_i h'$ implies $\sigma(h) = \sigma(h')$.
	
\end{definition}

We denote by $\Sigma_R$
the set of all memoryfull
uniform strategies.


A strategy for a coalition $A$ of agents is a set $\sigma_A
= \{ \sigma_a \mid a \in A\}$ of strategies, one for each agent in
$A$.
Given coalitions $A \subseteq B \subseteq Ag$, a strategy $\sigma_A$
for coalition $A$, a state $s \in S$, and a joint $B$-action $\vec b_B \in
Act^B$ that is enabled at $s$, we say that $\vec b_B$
is \emph{compatible with} $\sigma_A$ ({\em in} $s$)
whenever $\sigma_A(s) \sqsubseteq \vec
b_B$.
For states $s, s' \in S$ and strategy $\sigma_A$, we write
$s \xrightarrow{\sigma_A(s)} r$ if $s \xrightarrow{\vec a} r$ for
some joint action $\vec a \in Act^{Ag}$ that is compatible with
$\sigma_A$.

We define three notions of \emph{outcome} of strategy $\sigma_A$ at
history $h$, corresponding to the \emph{objective}, \emph{subjective},
and {\em common knowledge} interpretation of alternating-time
operators.  Fix a history $h$ and a strategy $\sigma_A$ for coalition
$A$.
\begin{enumerate}
	\item The set
	of \emph{objective outcomes of $\sigma_A$ at $h$}
	is defined as
	$out_{obj}\allowbreak (h, \sigma_A) = \big \{
	\lambda \in P \mid \lambda_{\leq |h|} = h \text{ and for all } j \geq |h|,\allowbreak
	\lambda[j] \xrightarrow{\sigma_A(\lambda_{\leq j})} \lambda[j + 1]\big\}
	$.
	
	\item The set
	of \emph{subjective outcomes of $\sigma_A$ at $h$} is defined as
	$
	out_{subj} \allowbreak (h, \sigma_A) =
	{\displaystyle \bigcup}_{i \in A, h' \sim_i h} out_{obj}(h', \sigma_A)$.
	
	\item The set
	of \emph{common knowledge (ck) outcomes of $\sigma_A$ at $h$} is defined as
	$out_{ck} (h,\allowbreak \sigma_A) =
	{\displaystyle \bigcup}_{h' \in C_A(h)} out_{obj} (h', \sigma_A)$.
\end{enumerate}

Intuitively, objective outcomes are paths beginning with the current
history $h$ and consistent with the current joint strategy $\sigma_A$;
whereas subjective (resp.~common knowledge) outcomes are paths
beginning with some history $h'$ {\em indistinguishable} from $h$
according to collective (resp.~common) knowledge (as well as
consistent with $\sigma_A$).
Again, notions of outcomes from states can be obtained from the
definitions above, as states are a particular type of histories.

\begin{definition} \label{semantics}
	Given an iCGS $\G$, a history formula $\phi$, path formula $\psi$, and
	$m \in \mathbb{N}$, the \emph{subjective} (resp.~\emph{objective}, {\em
		common knowledge}) satisfaction of $\phi$ at history $h$ and of
	$\psi$ in path $\lambda$, denoted $(\G, h) \models_{x} \phi$ and
	$(\G, \lambda,m) \models_{x} \!\psi$ for $x \in \{ subj,  obj, ck\}$,
	is defined recursively as
	follows (clauses for Boolean operators are immediate and thus omitted):
	{\small
		\begin{align*}
			& (\G, h)  \models_{x}  p  \! \! \!  \! \! \! &  \text{iff } &  p \in \pi(last(h)) \\
			&(\G, h) \models_{x} \dlangle A \drangle \psi & \text{iff } & \text{ for some } \sigma_{A} \in \Sigma_R,\\ & &  & \text{ for all } \lambda \in out_{x}(h, \sigma_A),
			(\G, \lambda, |h|) \models_{x} \psi\\
			&(\G, \lambda, m) \models_{x} \phi & \text{iff } & (\G, \lambda_{\leq m}) \models_{x} \phi\\
			&(\G, \lambda, m) \models_{x}  X \psi & \text{iff } &  (\G, \lambda, m+1) \models_{x}   \psi\\
			&(\G, \lambda, m) \models_{x}  \mathbf{Y} \psi & \text{iff } &  m \geq 1 \text{ and } (\G, \lambda, m-1) \models_{x}   \psi\\
			&(\G, \lambda, m) \models_{x}  \psi U \psi' & \text{iff } & \text{ for some } j \geq m, (\G, \lambda, j) \models_{x} \psi', \text{ and}\\
			& & & \text{for all } k,  m \leq k <j \text{ implies } (\G, \lambda, k) \models_{x} \psi
		\end{align*}
	}
\end{definition}

\begin{remark}
	The individual and common knowledge operators $K_i$ and $C_A$ of
	epistemic logic \cite{Fagin+95b} can be added to the syntax of ATL$^*$
	with the following (memoryful) interpretation:
	\begin{eqnarray*}
		(\G, h) \models_{x} K_i \phi & \text{iff} & \text{for all } h' \sim_i h,
		(\G, h') \models_{x} \phi\\ (\G, h)
		\models_{x} C_A \phi & \text{iff} & \text{for all } h'\in C_A(h),
		(\G, h') \models_{x} \phi
	\end{eqnarray*}
	
	Withnin the subjective or the  common knowledge
	interpretation of ATL$^*$, the individual knowledge operator becomes a derived operator, as we
	have $(\G,h) \models_{x} K_i \phi $ iff $(\G,h) \models_{x}
	\dlangle i \drangle \phi U \phi$ for both $x \in \{subj,ck\}$.  
	It is known that there exists no such definition for the knowledge
	operators in ATL$^*$ within the objective interpretation.
	Furthermore, and only for the case of the common knowledge interpretation, we may similarly derive the common knowledge operator as well: $(\G,h)
	\models_{ck} C_A \phi$ iff $(\G,h) \models_{ck} \dlangle A \drangle
	\phi U \phi$.
\end{remark}

%


\begin{example} \label{excoordination}
	
	We describe a coordination scenario comprising of two agents,
	$1$ and $2$, who have to agree on a meeting. But $1$ does not
	know where she is, in Paris or London, and therefore which is the time
	zone, while $2$ does not know if it is winter time or summer
	time. Agent $1$ can choose either go to the meeting ($g$) or wait one
	hour ($w$) whereas $2$ can choose either to go at 3pm ($3$) or
	at 4pm ($4$), local time.  Now suppose it is 3pm GMT.  In London, in
	the winter ($s_2$) 1 and 2 coordinate if 1 goes to the meeting
	and 2 goes at 3pm local time. They also meet if 1 waits one hour and 2
	goes at 4pm. All other combined actions are unsuccessful.  Analogously
	for Paris in the winter ($s_1$), and London in the summer ($s_3$).
	The iCGS $\G$ depicted in Fig.~\ref{fig:exmrev} shows the described
	scenario.  Since $1$ and $2$ have partial observability, $1$ (resp. $2$) cannot distinguish between states $s_2$ and $s_1$
	(resp.~$s_3$).  After the initial choice, $1$ and $2$ stay
	indefinitely in either $s_4$ or $s_5$.
	Finally, we use two atoms, to denote success ($s$) and failure ($f$), respectively.
	%
	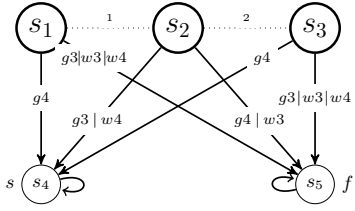
\begin{figure}
		\begin{center}
			\mbox{\scalebox{0.8}[0.8]{
					\begin{tikzpicture}
					[node distance = 6em]
					\node []
					(S0)
					{\Large $\stackrel{}{}$};
					\node []
					(S5)
					[below of = S0, node distance = 2em]
					{\Large $\stackrel{}{}$};
					\node [cnode]
					(S4)
					[left of = S5, node distance = 7em,label=left:{$s$}]
					{ $s_4$};
					\node [cnode]
					(S6)
					[right of = S5, node distance = 7em,label=right:{$f$}]
					{ $s_5$};
					\node [cnode, very thick]
					(S2)
					[above of = S0, node distance = 6em]
					{\Large $s_2$};
					\node [cnode, very thick]
					(S3)
					[right of = S2, node distance = 7em]
					{\Large $s_3$};
					\node [cnode, very thick]
					(S1)
					[left of = S2, node distance = 7em]
					{\Large $s_1$};
					
					\path[-stealth']
					(S2)
					edge  [pos = 0.6]
					node [] {\scriptsize $g4 \mid w3$}
					(S6)
					(S1)
					edge  [pos = 0.4]
					node [] {\scriptsize $g4$}
					(S4)
					edge  [pos = 0.13]
					node [] {\scriptsize $g3 \!\!\mid\!\! w3 \!\mid\!\! w4$}
					(S6)	
					(S2)	
					edge  [pos = 0.6]
					node [] {\scriptsize $g3 \mid w4$}
					(S4)
					(S3)	
					edge  [pos = 0.13]
					node [] {\scriptsize $g4$}
					(S4)
					edge  [pos = 0.4]
					node [] {\scriptsize $g3 \!\mid\! w3 \!\mid\! w4$}
					(S6)	
					(S4)	edge	[loop right]
					node [noall] {}
					()
					(S6)	edge	[loop left]
					node [noall] {}
					()
					;
					
					\draw[dotted] (S1) -- node[above] {\tiny{$1$}} (S2);
					
					\draw[dotted] (S2) -- node[above] {\tiny{$2$}} (S3);
					\end{tikzpicture}
				}
			}
			\caption{\label{fig:exmrev} The meeting scenario in Example~\ref{excoordination}.\vspace*{-10pt}}
		\end{center}
	\end{figure}
	
	As an example of specification in ATL$^*$, consider formula
	$\varphi = \dlangle \{ 1,2 \} \drangle X s$.
	This formula can be read as: $1$ and $2$ have a joint strategy
	to meet.
	Note that $\varphi$ is true in both $s_1$ and $s_3$ when considering
	the subjective interpretation.  However, is the truth of $\varphi$ in
	$s_1$ and $s_3$ justified from point of view of the rational behaviour
	of $1$ and $2$?
	Specifically, since $\varphi$ 
	is true in $s_1$ according to the subjective interpretation, both 1 and 2 know that they have a successful strategy, which consists in playing action $g$ for $1$ and action $4$ for $2$. But for this strategy to be successful (i.e., satisfying $X s$ for all outcomes) it assumes that $2$ is playing action $3$ in $s_2$: is such an assumption
	rationally justified?  Notice that in $s_2$,  $2$ considers state
	$s_3$ epistemically possible, and in $s_3$ the joint action $(g,3)$
	leads to failure. 
	Hence, it does not appear to be rational for $2$ to play
	$3$ in $s_2$.  Ever more so that, by playing $4$ in $s_2$ and $s_3$, $2$ can coordinate with 1 to achieve success.
	This example shows a scenario where, even though both agents know that
	their strategies are successful in principle, they do not necessarily
	coordinate, as they do not know that the other agent knows her
	strategy to be successful.  Indeed, we have that in both $s_1$ and
	$s_2$ formula $\varphi$ is false according to the common knowledge
	interpretation. So, it is not the case that they have common knowledge
	of their strategies being successful.

\end{example}



\section{Alternating Bisimulations} \label{altbisim}

In this section we introduce a notion of bisimulation suitable for
concurrent game structures with imperfect information.  In particular,
we show that it preserves the satisfaction of formulas in ATL$^*$,
when interpreted under imperfect information and perfect recall.
Firstly, we present several auxiliary notions.
Hereafter $\G = \langle Ag, S, s_0, Act, \{\sim_i \}_{i \in Ag},
d, \to, \allowbreak \pi \rangle$ and $\G' = \langle Ag, S', s'_0,
Act', \{\sim'_i \}_{i \in Ag}, d', \allowbreak \to', \pi'
\rangle$ are two iCGS defined on the same set $Ag$ of agents, with histories $h \in Hist(\G)$ and $h' \in Hist(\G')$.



\paragraph{Partial strategies.}
A \emph{partial (uniform, memoryful) strategy} for agent $i \in Ag$ is
a partial function $\sigma : H \rightharpoonup Act$ such that for each
$h, h' \in H$, $\sigma(h) \in d(i,last(h))$, and $h \sim_i h'$ implies
$\sigma(h) = \sigma(h')$.  We denote the domain of partial strategy
$\sigma$ as $dom(\sigma)$.  Given a coalition $A \subseteq Ag$,
a \emph{partial strategy for $A$} is a tuple $(\sigma_i)_{i \in A}$ of
partial strategies, one for each agent $i \in A$.  The set of partial
strategies for $A$ is denoted as $PStr_A$.
Given a set $Q \subseteq H$ of histories and coalition $A \subseteq Ag$,
we denote by $PStr_A(Q)$ the set of partial strategies whose domain is
$Q$:
\[
PStr_A(Q) = \big\{(\sigma_i)_{i \in A} \in PStr_A \mid dom(\sigma_i) =
Q \text{ for all } i \in A\big\}
\]

Additionally, given a (total or partial) strategy $\sigma_A$ and a
history $h \in dom(\sigma_A)$, define the set of \emph{successors of $h$
	by $\sigma$} as
\cata{
	\[
	succ(h, \sigma_A) = 
	\{ h \alpha s \mid \alpha \!\in\! Act^{Ag} \text{ with } \sigma_A(h) \sqsubseteq \alpha \text{ and } 
	h \xrightarrow{\alpha} s \}
	\]
}
%
Further, we set 
$succ(\sigma_A) =
\bigcup_{h \in dom(\sigma_A)}\allowbreak succ(h,\sigma_A)$.
%
%
%
\begin{definition}[Strategy simulators]
	Given a coalition $A \subseteq Ag$, an \emph{$A$-strategy simulator}
	(or simply \emph{strategy simulator}, when $A$ is understood from the
	context) is a family $ST
	= \allowbreak \big(ST_{C_A(h),C_A(h')}\big)_{h \in Hist(\G), h' \in
		Hist(\G')}$ of  mappings $ST_{C_A(h),C_A(h')}:
	PStr_A(C_A(h)) \rightarrow PStr_A(C_A(h'))$
	such that for all histories $h,k \in Hist(\G)$ and $h',k' \in
	Hist(\G')$,
	\begin{multline}
		\text{if } C_A(h) = C_A(k) \text{ and } C_A'(h') = C_A'(k') \\
		\text{ then } ST_{C_A(h),C_A'(h')} = ST_{C_A(k),C_A'(k')}
		\label{prop:ST-independence}
	\end{multline}
\end{definition}
Hereafter, we simplify the notation by writing $ST(\sigma)$ instead of
the cumbersome $ST_{C_A(h),C_A(h')}(\sigma)$, whenever $h$ and $h'$ are
clear from the context and $\sigma \in PStr(C_A(h))$.


We can now introduce the notion of (bi)simulation for iCGS.
\begin{definition}[Memoryful Simulation] \label{def:bisim}
	Let $A\subseteq Ag$ be a coalition of agents.  A relation
	$\altsim{A} \subseteq Hist(\G) \times Hist(\G')$ is a \emph{simulation for $A$} iff
	there exists a strategy simulator $ST$ such that for any two histories $h \in Hist(\G)$, $h' \in Hist(\G')$,
	$h \altsim{A} h'$ implies the following:\label{it:original}
	\begin{enumerate}
		\item\label{it:agree} $\pi(last(h)) = \pi'(last(h'))$;
		
		\item For every $i \in A$ and $k' \in Hist(\G')$, if $h' \sim'_i k'$ then for some
		$k \in Hist(\G)$, $h \sim_i k$ and $k \altsim{A} k'$.

		\item 
		For every pair of histories $k \in C_A(h)$ and $k' \in C_A'(h')$ with $k \!\altsim{A}\! k'\!$,
		for every partial strategy $\sigma_A \!\!\in\! PStr_A(C_A(h)\!)$ 
		and every history $l' \in succ(k', ST(\sigma_A))$,
		there exist a history $l \in succ(k, \sigma_A)$ such that $l \altsim{A} l'$.
	\end{enumerate}
	
	A relation $\altbisim{A}$ is a {\em bisimulation} iff both
	$\altbisim{A}$ and its converse $\altbisim{A}^{-1} = \{ (h',h) \mid
	h \altbisim{A} h' \}$ are simulations.
\end{definition}


We also extend (bi)-simulation to paths $\lambda \in Path(\G)$, $\lambda' \in Path(\G')$, by denoting 
$\lambda \altsim{A} \lambda'$ 
iff for all $j \geq 0$, $\lambda_{\leq j} \altsim{A} \lambda'_{\leq j}$.

The main result of this section, Theorem~\ref{thm-bisim-mem}, shows
that bisimilar iCGS satisfy the same formulas in ATL$^*$ under
imperfect information and perfect recall. To prove this result, we
need the following auxiliary lemma:
\begin{lemma}\label{lemma:infinite-strats2}
	If $h \altsim{A} h'$ then
	for every strategy $\sigma_A$, there exists a
	strategy $\sigma'_A$ such that
	\begin{itemize}
		\item[($*$)] for every path $\lambda' \in  out_x(h',\sigma_A')$, for $x \in \{subj,obj, ck\}$,
		there exists a path $\lambda \in out_x(h,\sigma_A) $ such
		that $\lambda \altsim{A} \lambda'$.
	\end{itemize}
\end{lemma}
\begin{proof}
	First, notice that point 3 in Def.~\ref{def:bisim} can be rewritten
	as:
	\begin{enumerate}
		\item[\^3.] 
		For all histories $k \in C_A(h)$ and $k' \in C_A'(h')$ such that
		$k \altsim{A} k'$, for all partial strategies $\sigma_A \in
		PStr_A(C_A(h))$, there exists a mapping $\rho_{\sigma_A,k,k'} :
		succ(k',ST(\sigma_A)) \rightarrow succ(k,\sigma_A)$ such that for
		all histories $l'\!\!\in\! succ(k'\!, ST(\sigma_A))$,
		$\rho_{\sigma_A,k,k'}(l') \altsim{A} l'$.
	\end{enumerate}
	\cata{
		in which the mapping $\rho_{\sigma_A,k,k'} $ represents the \emph{skolemization}, in the original point 3, 
		of the existential quantifier over $l\in succ(k,\sigma_A)$,
		seen as a unary function on $l'\in succ(k',ST(\sigma_A))$ indexed by 
		$\sigma_A \in PStr_A(C_A(h))$, $k\in C_A(h)$ and $k'\in C_A'(h')$.}
	
	We now define the sequence $\big(dom^n(\sigma_A)\big)_{n \in \mathbb{N}}$,
	of sets of histories in $\G$ such that $k \in dom^n(\sigma_A)$ iff $k$
	can be reached in at most $n$ steps from $C_A(h)$ by applying actions
	compatible with strategy $\sigma_A$. Formally, $dom^0(\sigma_A) =  C_A(h)$ and 
	\[	 
	dom^{n+1}(\sigma_A)=  dom^n(\sigma_A) \cup \!\!\!\! \bigcup_{k \in dom^n(\sigma_A)} \!\!\!\! \big\{ C_A(l) \mid l \in succ(k, \sigma_A)\big\}
	\]
	
	Also, we denote by $\sigma^n_A$ the partial strategy resulting from
	restricting $\sigma_A$ to $dom^n(\sigma_A)$.
	
	We then define inductively
	a sequence $(\br \sigma^n_A)_{n \in \mathbb{N}}$ of partial strategies
	in $\G'$
	such that $dom(\br \sigma^n_A) \subseteq dom(\br\sigma^{n+1}_A)$ for every
	$n \in \mathbb{N}$, and, at the same time, a sequence of mappings 
	$\theta^n_A : dom(\br \sigma^n_A) \rightarrow dom^n(\sigma_A)$,
	satisfying the following property:
	\begin{equation}\label{id-theta-sigma-br}
		\theta^{n+1}_A(k') \in succ(\theta^n_A(k'_{\leq |k'|-1}),\sigma_A)
	\end{equation}
	
	%
	
	\noindent 
	The sequences 
	$(\br \sigma^n_A)_{n\in \mathbb{N}}$ and $(\theta^n_A)_{n\in \mathbb{N}}$ are defined as follows:
	\begin{enumerate}
		\item
		{\small	$dom(\br \sigma_A^0) = C'_A(h')$;}
		
		{\small	$dom(\br \sigma_A^{n+1}) =
			dom(\br \sigma_A^n) \cup \bigcup \big\{
			succ(k',\br \sigma_A^n) \mid k' \in dom(\br \sigma^n_A)\big \}$.}
		
		\item For all $k' \in dom (\br \sigma_A^0)$,
		$\br \sigma_A^0(k') = ST(\sigma^0_A)(k')$.
		
		\item For all $k' \in C'_A(h')$, we fix a unique $k\in C_A(h)$ such that $k \altsim{A} k'$ (which exists by point 2 in Def.~\ref{def:bisim}),
		and define $\theta^0_A(k') = k$.
		
		\item For all $k' \in dom(\br \sigma_A^{n+1})$, let $l' =
		k'_{\leq |k'|-1}$. Then, we set
		$\theta^{n+1}_A(k') = \rho_{\sigma_A,\theta^n_A(l'),l'}(k')$.
		\item For all $k' \in dom(\br \sigma_A^{n+1})$,\\
		$\br\sigma^{n+1}_A(k') = \big(ST_{C_A(\theta^n_A(k')),C_A'(k')}\big)(\sigma_A\restr{C_A(\theta^n_A(k'))})(k')$.
	\end{enumerate} 
	
	We prove property~(\ref{id-theta-sigma-br}) above, as well as the following:
	for every $k'\in dom(\br \sigma^n_A)$,
	\begin{align*}
		\theta^n_A(k') & \altsim{A} k' \quad (*) & 
		\theta^n_A(k') & \in dom^n(\sigma_A) \qquad (**)
	\end{align*}
	
	Property $(*)$
	holds by definition, since
	$\rho_{\sigma_A,\theta^n_A(l'),l'}(k') \altsim{A} k'$.
	%
	Property~(\ref{id-theta-sigma-br}) and $(**)$
	can be proved
	by induction on $n = |k'| - |h'|$, by observing that $(**)$
	holds for $n=0$; if (\ref{id-theta-sigma-br}) holds for $n+1$ then $(**)$
	holds for $n+1$ too; and finally, (\ref{id-theta-sigma-br})
	is an immediate consequence of the definition of $\theta^n_A$,
	property $(**)$
	Def.~\ref{def:bisim}. Note that property $(**)$
	ensures
	that indeed $\theta^n_A : dom(\br \sigma^n_A) \rightarrow
	dom^n(\sigma_A)$ as desired.
	
	The "limit" of the sequence of strategies
	$(\br \sigma^n_A)_{n\in \mathbb{N}}$
	is still a partial strategy and the domain of each $\br\sigma^n_A$ might not be closed under the
	common knowledge indistinguishability relation $\sim^C_A$. 
	So, we extend first the domain of each $\br\sigma^n_A$ to one
	which is closed under $\sim^C_A$ in $\G'$, by constructing
	the sequence of partial strategies
	$(\hat \sigma^n_A)_{n\in \mathbb{N}}$ and the sequence of mappings
	$\hat\theta^n_A : dom(\hat \sigma^n_A) \rightarrow dom^n(\sigma_A)$, as follows:
	\begin{enumerate}
		\item
		$dom(\hat \sigma_A^0) =dom(\br \sigma^0_A) = C'_A(h')$;
		
		$dom(\hat \sigma_A^{n+1}) =
		dom(\br \sigma_A^n) \cup \bigcup \big\{ C_A'(l') \mid 
		\exists k' \in dom(\hat \sigma^n_A) \\
		\hspace*{50pt} \text{ with } l\in succ(k',\br \sigma_A^n) \cap C_A'(k') \big \}$.
		
		\item For all $k' \in dom (\hat \sigma_A^0)$,
		$\hat \sigma_A^0(k') = ST(\sigma^0_A)(k')$.
		
		\item For all $k' \in C'_A(h')$, $\hat\theta^0_A(k') = \theta^0_A(k')$.
		
		\item For all $k' \in dom(\hat \sigma_A^{n+1})$, let $l' =
		k'_{\leq |k'|-1}$. Then, we set 
		$\hat\theta^{n+1}_A(k') = \rho_{\sigma_A,\theta^n_A(l'),l'}(k')$.
		\item For all $k' \in dom(\hat \sigma_A^{n+1})$,\\
		$\hat\sigma^{n+1}_A(k') = \big(ST_{C_A(\theta^n_A(k')),C_A'(k')}\big)(\sigma_A\restr{C_A(\theta^n_A(k'))})(k')$.
	\end{enumerate} 
	
	We
	observe that properties $(*)$ and $(**)$
	still hold for
	$\hat \sigma^n_A$ and $\hat \theta^n_A$, though property
	(\ref{id-theta-sigma-br}) does not in general.  In this way we get
	that $dom(\hat \sigma^n_A) \supseteq dom(\br \sigma^n_A)$ and for every
	$k' \in dom(\br \sigma^n_A)$, $n \in \mathbb{N}$,
	$\hat \sigma^n_A(k')
	= \br \sigma^n_A(k')$.
	As a result, the ``limit'' partial strategy $\hat \sigma_A =
	{\displaystyle\bigcup_{n\in \mathbb{N}}} \hat \sigma^n_A$ defined as
	$\hat\sigma_A(k') = \hat \sigma^{|k'|-|h'|}_A(k')$ is also uniform and
	its domain $dom(\hat \sigma_A)$ is closed under $\sim^C_A$.  We then
	transform it into a (total) uniform strategy $\sigma_A'$ by imposing
	a fixed action $a_0 \in Act$ wherever $\hat \sigma^n_A$ is undefined,
	that is,
	$\sigma_A'(k') = \hat \sigma_A(k')$ for $k' \in dom(\hat \sigma_A)$ and 
	$\sigma_A'(k') = a_0$ otherwise.
	
	
	Finally, to prove property ($*$) for the common knowledge semantics,
	consider a path $\lambda' \!\in\! out^{\G'}_{ck}(h',\sigma_A') $ and the
	sequence $\big(\theta^n_A(\lambda'_{\leq
		|h'|+n})\big)_{n\in \mathbb{N}}$ of histories in $\G$.  By
	construction, $\theta^{n+1}_A(\lambda'_{\leq |h'|+n+1}) \!\in\!
	succ(\theta^n_A(\lambda'_{\leq |h'|+n}),\sigma_A)$ and
	$\theta^n_A(\lambda'_{\leq |h'|+n}) \!\altsim{A} \lambda'_{\leq |h'|+n}$,
	which means that this sequence of histories is in fact a path
	$\lambda$ in $\G$ which is compatible with $\sigma_A$ and satisfies
	$\lambda \altsim{A} \lambda'$, which ends the
	proof.
\end{proof}


%

By using Lemma~\ref{lemma:infinite-strats2} we are finally able to
prove the main preservation result of this paper.


%
\begin{theorem}\label{thm-bisim-mem}
	Let
	$h \in Hist(\G)$ and $h'\in Hist(\G')$ be histories such that
	$h \altbisim{A} h'$, and $\lambda \in Path(\G)$ and $\lambda' \in
	Path(\G')$ be paths such that $\lambda \altbisim{A} \lambda'$.
	Then, for every history $A$-formula $\phi$, path $A$-formula $\psi$,
	$m \in \mathbb{N}$, and $x \in \{ subj,  obj, ck\}$,
	\begin{eqnarray*}
		(\G,h) \models_x \phi & \text{iff} & (\G',h') \models_x \phi\\
		(\G,\lambda, m) \models_x \psi & \text{iff} & (\G',\lambda', m) \models_x \psi
	\end{eqnarray*}
\end{theorem}
%
\begin{proof}
	The proof is by mutual induction on the structure of $\phi$ and
	$\psi$.
	
	The case for propositional atoms is immediate as for $x \in \{ subj,
	obj, ck\}$, $(\G, h) \!\models_x\! p$ iff $p \!\in\! \pi(last(h))$, iff
	$p \!\in\! \pi'(last(h')\!)$ by item 1 in Def.~\ref{def:bisim}, iff
	$(\G',h') \models_x p$.  The inductive cases for propositional
	connectives are also immediate.
	
	For $\psi = \phi$, suppose that $(\G, \lambda, m) \models_x \psi$, that is,
	$(\G, \lambda_{\leq m}) \models_x \phi$.  By assumption,
	$\lambda_{\leq m} \altbisim{A} \lambda'_{\leq m}$ as well, and by induction
	hypothesis $(\G', \lambda'_{\leq m}) \models_x \phi$. Thus,
	$(\G', \lambda', m) \models_x \psi$.
	
	For $\psi = X \psi'$, suppose that 
	$(\G, \lambda, m+1) \models_x \psi'$.
	By the induction hypothesis, $(\G', \lambda', m+1) \models_x \psi'$. Thus,
	$(\G', \lambda', m) \models_x \psi$.  The inductive cases for $\psi = \mathbf{Y} \psi'$ and
	$\psi = \psi'
	U \psi''$
	is similar.
	
	Finally, for $\phi = \dlangle A \drangle \psi$, $(\G, h) \models_x \phi$
	iff for some strategy $\sigma_{A}$, for all $\lambda \in
	out^{\G}_x(h, \sigma_A)$, $(\G, \lambda, |h|) \models_x \psi$.  By
	Lemma.~\ref{lemma:infinite-strats2}, there exists stategy $\sigma'_A$
	s.t.~for all $\lambda' \in out^{\G'}_x(h', \sigma'_A)$, there exists
	$\lambda \in out^{\G}_x(h, \sigma_A)$
	s.t.~$\lambda \altbisim{A} \lambda'$.  Since $|h| = |h'|$, by the
	induction hypothesis $(\G, \lambda, |h|) \models_x \psi$ iff
	$(\G', \lambda', |h'|) \models_x \psi$.  Hence, $(\G',
	h') \models_x \phi$.
\end{proof}

\section{Bisimulations Games} \label{bisimgames}


In this section we introduce bisimulation games played on two iCGS and
we prove that the existence of a winning strategy for the \duplic{}
coalition is equivalent to the existence of a bisimulation between the
iCGS.
%
\begin{definition}[Bisimulation Game] \label{def:3bisim-games}
	Given iCGSs $\G$ and $\G'$, defined on the same sets $Ag$ of agents
	and $AP$ of atoms, a relation $R \subseteq Hist(\G) \times Hist(\G')$,
	and a pair $(h_0,h_0') \in Hist(\G) \times Hist(\G')$ of histories, we
	define the bisimulation game $\B(\G,\G',R,h_0,h_0')$ as a turn-based game of
	imperfect information between \emph{four} players: 
	\pdupl,\pspoil, called P-players, and \idupl, \ispoil, called I-players,
	organized in two coalitions: 
	the \duplic{} coalition $\{$\pdupl, \idupl$\}$ and
	the \spoiler{} coaltion $\{$\pspoil, \ispoil$\}$, 
	with both P-players having perfect information while 
	both I-players have \textbf{the same} imperfect information.
	
	\cata{
		At a higher-level, the bisimulation game is a turn-based game in which
		the I-players are in charge of defining the strategy simulators, in
		the sense that \ispoil{} chooses a partial strategy for $A$ over some
		common knowledge neighbourhood in one of the game structures,
		and \idupl{} responds with an appropriate partial strategy for $A$ in
		the other game structure.  Then the perfectly-informed players come
		into play, by appropriately defining mappings between histories
		compatible with the chosen strategies, which represent
		"skolemizations" of conditions $(2)$ and $(3)$ in
		Def.~\ref{def:bisim}.}
	
	\cata{
		The necessity for \ispoil{} and \idupl{} to only have imperfect information comes from the fact 
		that the same strategy profile has to be chosen by both players at positions which belong to the same common knowledge neighborhood in both game structures,
		since perfect information might be used by each player to trick the other player by choosing 
		a strategy which is not uniform for some agent in coalition $A$.
	}
	
	More formally, the game proceeds as follows:
	
	\begin{enumerate}
		\setcounter{enumi}{-1}
		\item The positions of the game form a \textbf{labeled tree}, denoted $T(\B)$, with the root  position labeled $(h_0,h_0')$.
		The rest of positions and their labels are given below.
		
		\item Each position $(h,h')$ where $\pi(h) \neq \pi'(h')$ or
		$(h,h') \not \in R$ is \textbf{winning} for the \spoiler{} coalition.
		
		\item Each position labeled $(h,h') \in Hist(\G) \times Hist(\G')$ belongs to \ispoil, 
		and both I-players receive observation $C_A(h) \times C_A'(h')$. 
		In each such position \ispoil{} may choose between two types of transitions:
		\begin{enumerate}
			\item For each $\sigma_A \in PStr(C_A(h))$, a transition to a successor (of the current position in the tree) 
			labeled $(h,h',\sigma_A,L)$. 
			\item For each $\sigma_A' \!\in\! PStr(\!C_A'(h')\!)$ a transition to a successor labeled $(\!h,h'\!\!,\sigma_A',\!R)$. 
		\end{enumerate}
		
		\item Each position $(h,h',\sigma_A,L)$ belongs to \idupl{} and both I-players observe $\sigma_A$.
		\idupl{} may choose, for each $\sigma_A' \in PStr(C_A'(h'))$, a transition to a 
		successor labeled $(h,h',\\ \sigma_A,
		\sigma_A',L)$.
		\item Each position $(h,h',\sigma_A',R)$ belongs to \idupl{} and both I-players observe $\sigma_A'$.
		\idupl{} may choose, for each $\sigma_A \in PStr(C_A(h))$, a transition to a 
		successor labeled $(h,h',\\ \sigma_A,
		\sigma_A',R)$.
		
		\item Each position $(h,h',\sigma_A,\sigma_A',L)$ belongs to \pspoil\
		and \pspoil{} may choose, for each $k' \in C_A'(h')$, a transition to 
		a successor labeled $(h,h',\sigma_A,\sigma_A',k',L)$.
		
		In all positions at points 5-12,  both I-players observe $C_A(h)\times C_A'(h')$
		\item Each position $(h,h',\sigma_A,\sigma_A',R)$ belongs to \pspoil, 
		and \pspoil{} may choose, for each $k \in C_A(h)$, a transition to 
		a successor labeled $(h,h',\sigma_A,\sigma_A',k,R)$.
		
		\item Each position $(h,h',\sigma_A,\sigma_A',k',L)$ belongs to \pdupl, 
		and \pdupl{} may choose, for each $k \in C_A(h)$, a transition to 
		a successor labeled $(h,h',\sigma_A,\sigma_A',\allowbreak k,k',L)$.
		\item Each position $(h,h',\sigma_A,\sigma_A',k,R)$ belongs to \pdupl, 
		and \pdupl{} may choose, for each $k' \in C_A'(h')$, a transition to 
		a successor labeled $(h,h',\sigma_A,\sigma_A',\allowbreak k,k',R)$.
		
		\item Each position $(h,h',\sigma_A,\sigma_A',k,k',L)$ belongs to \pspoil, 
		This position is \textbf{winning} for the \spoiler{} coalition if $\pi(k) \neq \pi'(k')$ or
		$(k,k') \not \in R$.
		In this position \pspoil{} may choose, for each $l' \in succ(k',\sigma_A')$, a transition to 
		a successor labeled $(h,h',\sigma_A,\sigma_A',k,k',l',L)$.
		\item Each position $(h,h',\sigma_A,\sigma_A',k,k',R)$ belongs to \pspoil, 
		This position is \textbf{winning} for the \spoiler{} coalition if $\pi(k) \neq \pi'(k')$ or
		$(k,k') \not \in R$.
		In this position \pspoil{} may choose, for each $l \in succ(k,\sigma_A)$, a transition to 
		a successor labeled $(h,h',\sigma_A,\sigma_A',k,k',l,R)$.
		
		\item Each position $(h,h',\sigma_A,\sigma_A',k,k',l',L)$ belongs to \pdupl, 
		and \pdupl{} may choose, for each $l \in succ(k,\sigma_A)$, a transition to 
		a successor labeled $(l,l')$ from where Rule 1 above applies.
		\item Each position $(h,h',\sigma_A,\sigma_A',k,k',l,R)$ belongs to \pdupl, 
		and \pdupl{} may choose, for each $l' \in succ(k',\sigma_A')$, a transition to 
		a successor labeled $(l,l')$ from where Rule 1 above applies.
	\end{enumerate}
\end{definition}


In the sequel, given a position $p \in T(\B)$, we denote $Obs(p)$ the
set of positions which give the same observation as $p$ to any of the
I-players.  Also, the set of strategies for the \duplic{} (resp.~\spoiler{}) coalition is
denoted $\Sigma_{Dupl}$ (resp.~$\Sigma_{Spoil}$).
Further, the set of positions which are compatible with a strategy
$\sigma \in \Sigma_{Dupl}\cup \Sigma_{Spoil}$ is denoted
$Comp(\sigma)$.  Finally, for each position $p$ we denote $lab(p)$ its
label, as per Def.~\ref{def:3bisim-games} of
bisimulation game.


Next, we prove that bisimulation relations and bisimulation games are
equivalent characterisations of iCGS.  To this end, given a history
$h_0\in Hist(\G)$, we define the \emph{pointed} iCGS $\G(h_0)$ in
which the initial state is $h_0$ 
and the transitions are modified accordingly. 

\begin{theorem} \label{game}
	For any $A$-bisimulation relation $R$ between $\G(h_0)$ and $\G'(h_0)$ 
	the \duplic{} coalition 
	has a strategy to win the bisimulation game $\B(\G\!, \G'\!\!,R,h_0,\! h'_0)$.
	
	Conversely, if the \duplic{} coalition has a joint strategy
	$\sigma_D$ to win the game $\B(\G, \G',R,h_0,h'_0)$, then
	there exists an $A$-bisimulation $\altbisim{A}$ with $\altbisim{A}
	\subseteq R \cap \big\{(h,h') \mid (h,h') \in
	out\big(p_{h_0,h_0'},\sigma_D\big) \big\}$,
	where $p_{h_0,h_0'}$ is the initial position of the bisimulation game $\B(\G, \G',R,h_0,h'_0)$.  
\end{theorem}
\begin{proof}
	We prove this theorem by double inclusion.
	
	$\Rightarrow$ Suppose that $\altbisim{A}$ is an $A$-bisimulation.
	For convenience, we utilize, as in the proof of Lemma
	\ref{lemma:infinite-strats2}, the restated variant (\^3) of point
	(3) in Def.~\ref{def:bisim} of $A$-simulations, which assumes a
	mapping $\rho_{\sigma_A,k,k'} : succ(k',ST(\sigma_A))\rightarrow
	succ(k,\sigma_A)$ that ensures that for any $l' \in
	succ(k',ST(\sigma_A))$, we have $\rho_{\sigma_A,k,k'}(l')
	\altbisim{A} l'$.  Since $\altbisim{A}$ is also a reverse
	simulation, we symmetrically consider $\rho'_{\sigma_A',k,k'} :
	succ(k,ST'(\sigma_A'))\rightarrow succ(k',\sigma_A')$
	s.t.~$\rho'_{\sigma_A',k,k'}(l') \!\altbisim{A}\! l$ for any $l \!\in\!
	succ(k,ST'(\sigma_A')\!)$.
	
	Similarly, we restate point (2) in Def.~\ref{def:bisim} in functional
	terms: \def\thetalar{\theta^{\leftarrow}}
	\def\thetarar{\theta^{\rightarrow}}
	\begin{enumerate}
		\item [\^2] For for each $\sigma_A \in PStr(C_A(h))$ there exists a mapping $\thetalar_{\sigma_A} : C_A'(h') \rightarrow
		C_A(h)$ such that for any $i\in A$, if $k_1'\sim_i' k_2'$ then
		$\thetalar_{\sigma_A}(k_1') \sim_i \thetalar_{\sigma_A}(k_2')$.
	\end{enumerate}
	\cata{ To see that this formulation is equivalent to item (2) in
		Def.~\ref{def:bisim}, note first that this point restates as the
		first-order formula $\varphi = \forall k' \in Hist(\G') \exists k
		\in Hist(\G) \Big( h' \sim'_i k' \rightarrow h \sim_i k \wedge k
		\altsim{A} k'\Big)$.  This formula is equivalent to $\forall
		\sigma_A.(\varphi \wedge \sigma_A \in PStr(C_A(h))) $ by the
		Universal Generalization Rule since $\sigma_A$ is not free in
		$\varphi$. Then $\thetalar_{\sigma_A} : C_A'(h') \rightarrow C_A(h)$
		corresponds to the \emph{skolemization} of $\exists k \in Hist(\G)$
		(seen as a unary function indexed by $\sigma_A$).  }
	
	
	By symmetry, for each $\sigma_A' \in PStr(C_A'(h'))$ we denote
	$\thetarar_{\sigma_A'} :C_A(h) \rightarrow C_A'(h')$ the reverse
	mapping, which exists since $\altbisim{A}$ is also a (reverse)
	simulation between $\G'$ and $\G$.
	
	Then, we define the strategy profile $(\sigma_{ID},\sigma_{PD})$ for
	the \duplic{} coalition as follows: for any position $p$,
	\begin{enumerate}
		\item If $lab(p) = (h,h',\sigma_A,L)$ then $\sigma_{ID}(p) =
		ST(\sigma_A)$, and if $lab(p) = (h,h',\sigma_A',R)$ then
		$\sigma_{ID}(p) = ST'(\sigma_A')$.
		\item If $lab(p) = (h,h',\sigma_A,\sigma_A',k',L)$ then $\sigma_{PD}(p) = \thetalar_{\sigma_A}(k')$,
		and if $lab(p) \!=\! (h,h',\allowbreak\sigma_A,\sigma_A',k,R)$ then $\sigma_{PD}(p) \!=\! \thetarar_{\sigma_A'}(k)$.
		\item If $lab(p) = (h,h',\sigma_A,\sigma_A',k,k',l',L)$ then $\sigma_{PD}(p) = 
		\rho_{\sigma_A,k,k'}(l')$ and 
		if $lab(p) = (h,h',\sigma_A,\sigma_A',k,k',l,R) \\$ then $\sigma_{PD}(p) = \rho'_{\sigma_A',k,k'}(l)$.
	\end{enumerate}
	
	Since $\altbisim{A}$ is an $A$-bisimulation and $ST$, $ST'$ are
	strategy simulators that do not depend on $h$ or $h'$, strategy
	$\sigma_{ID}$ is uniform, that is, for all positions $p,p'$ belonging
	to \idupl{} and this player receives the same sequence of observations 
	along the history that leads to $p$ and the history that leads to $p'$,
	we must have   $\sigma_{ID}(p)= \sigma_{ID}(p')$.   
	Then, all the
	runs that are compatible with the strategy profile $\sigma_D$ never
	reach a position $(h,h')$ where \spoiler{} wins:
	\begin{enumerate}
		\item[a.] For $lab(p) = (h,h',\sigma_A,\sigma_A',k',L)$,
		$lab(succ(p,\sigma_D)) =
		(h,h'\!,\sigma_A,\sigma_A',\thetalar_{\sigma_A}(k'),k',\allowbreak L)$.  But
		$\thetalar_{\sigma_A}(k')\!\altsim{A}\! k$ by
		point
		(\^2)
		for Def.~\ref{def:bisim}, which implies that
		$succ(p,\sigma_D)$ is not winning for \spoiler{}.  A similar
		argument holds for $lab(p) = (h,h',\sigma_A,\allowbreak\sigma_A', k,R)$.
		
		\item[b.] For $lab(p) \!=
		\!(h,h'\!,\sigma_A,\sigma_A',k,k'\!,l'\!,L)$,
		$lab(succ(p,\sigma_D)) \\=\!
		(h,h'\!,\sigma_A,\sigma_A',k,k',\rho_{\sigma_A,k,k'}(l'),l',L) $.
		But $\rho_{\sigma_A,k,k'}(l') \\ \altsim{A} l'$ by
		point
		(\^3) for Def.~\ref{def:bisim}, which implies that
		$succ(p,\\ \sigma_D)$ is not winning for \spoiler{}.  A similar
		argument holds for $lab(p) = (h,h',\sigma_A,\sigma_A',k,k',l,R)$.
	\end{enumerate}
	
	$\Leftarrow$ Suppose now that we have a winning joint strategy
	$\sigma_D = (\sigma_{ID},\sigma_{PD})$ for the \duplic{} coalition.
	Then, for each position $p$ that is consistent with $\sigma_D$, with
	label $lab(p) = (h,h') \in Hist(\G)\times Hist(\G')$, we set $h
	\altbisim{A}^{\sigma_D} h'$.
	
	The strategy simulators are then defined as follows: for each $h
	\altbisim{A}^{\sigma_D} h'$ with $(h,h') = lab(p)$ for some position
	$p$ in the bisimulation game, and each $\sigma_A \in PStr(C_A(h))$,
	note first that we have a \ispoil{} transition to a position
	$p^1_{\sigma_A}$ labeled $(h,h',\sigma_A,L)$ and then a \idupl{}
	transition to a position $p^2_{\sigma_A}$ labeled
	$(h,h',\sigma_A,\sigma_{ID}(p^1_{\sigma_A}),L)$.  Then, we set
	$ST_{C_A(h),C_A'(h')}(\sigma_A) = \sigma_{ID}(p^1_{\sigma_A})$.  Note
	that this definition is independent of the choice of $p$ since, by
	construction, all positions $\br p$ with $lab(\br p) = lab(p)$ are
	indistinguishable for \idupl{}, as he observes only $C_A(h)\times
	C_A'(h')$ and $\sigma_A$. Hence, $\sigma_{ID}(p^1_{\sigma_A}) =
	\sigma_{ID}({\br p}^1_{\sigma_A})$, where ${\br p}^1_{\sigma_A}$ is
	the position resulting by \ispoil{} choosing $\sigma_A$ in position
	$\br p$.  This ensures that $ST$ is indeed a strategy simulator
	according to Equation~(\ref{prop:ST-independence}).
	
	
	Furthermore, the mappings $\thetalar_{\sigma_A}$ are defined as
	follows:
	given position $p$ with $lab(p) = (h,h')$ as above, then for each $\sigma_A \in
	PStr(C_A(h))$, denote first $p^1_{\sigma_A}$ the position  
	resulting from \ispoil{}
	executing transition $\sigma_A$;  further denote 
	$p^2_{\sigma_A}$ the position  which belongs to \pspoil{} 
	after \pdupl{} executes action $\sigma_D(p^1_{\sigma_A})$. 
	Note then that, in position $p^2_{\sigma_A}$, for each $k'
	\in C_A'(h')$, \pspoil{} has a move to a position $p^3_{\sigma_A,k'}$
	which is labeled $(h,h',\sigma_A,\sigma_{ID}(p^1_{\sigma_A}),k',L)$
	which belongs to \pdupl{}.  Then we define $\thetalar_{\sigma_A}(k') =
	\sigma_{PD}(p^3_{\sigma_A,k'})$.
	
	This definition is dependent on the choice of the starting position $p$, but this
	is not an issue for our definition of 
	$\altbisim{A}^{\sigma_D}$ since there is no requirement for building
	the \emph{maximal} bisimulation associated with a bisimulation game.
	Note further that this definition, together with the fact that
	$\sigma_D$ is winning and hence position $p^4_{\sigma_A,k'} =
	succ(p^3_{\sigma_A,k'},\sigma_{PD}(p^3_{\sigma_A,k'}))$ is not winning
	for the \spoiler{} coalition, implies that $\pi(\thetalar_{\sigma_A}(k')) =
	\pi'(k')$ and further ensures that $\thetalar_{\sigma_A}(k')$
	satisfies the restated point (\^2) for Def.~\ref{def:bisim}.
	
	Finally, by proceeding from position $p^4_{\sigma_A,k'}$, which again
	belongs to \pspoil{}, for each \pspoil{}'s choice of some history $l' \in
	succ(k',\sigma_{ID}(p^1_{\sigma_A}))$, the game proceeds to a position
	$p^5_{\sigma_A,k',l'}$ that belongs to \pdupl{} and is labeled with
	$lab(\allowbreak p^5_{\sigma_A,k',l'}) \!=\!
	(h,h'\!,\sigma_A,\sigma_{ID}(p^1_{\sigma_A}),\sigma_{PD}(p^3_{\sigma_A,k'}),k'\!,l'\!,L)$.
	We then define $\rho_{\sigma_A,\thetalar_{\sigma_A}(k'),k'}\allowbreak(l') =
	\sigma_{PD}(p^5_{\sigma_A,k',l'})$.  Also note that this definition is
	dependent on the choice of the initial position $p$ with no loss of
	generality, and the fact that $\sigma_D$ is winning ensures that the
	position resulting from $p^5_{\sigma_A,k',l'}$ by \pdupl{}'s choice
	and labeled $(\sigma_{PD}(p^5_{\sigma_A,k',l'},l')$, is not winning
	for \spoiler{}. In particular, $\pi(p^5_{\sigma_A,k',l'}) =
	\pi'(l')$ and
	$\rho_{\sigma_A,\thetalar_{\sigma_A}(k'),k'}$ satisfies the restated
	point (\^3) for Def.~\ref{def:bisim}.
	
	Similar considerations give us the definitions for
	$\thetarar_{\sigma_A'}$ and $\rho'_{\sigma_A',k,k'}$ for each
	$\sigma_A' \!\in\! PStr(C_A'(h'))$, $k'\!\in\! C_A'(h')$ and $k\!\in\!
	C_A(h)$ with $k\altbisim{A} k'$.  This completes the proof of
	Theorem \ref{game}.
\end{proof}

We conclude this section with some immediate properties about our
bisimulation games and bisimulation relations.
\begin{proposition}\label{prop:tot}
	\begin{enumerate}
		\item The set of bisimulations associated with the same strategy simulator forms a complete lattice w.r.t.~set inclusion.
		\item If two iCGS $\G$ and $\G'$ are bisimilar, then the \duplic{} coalition has a winning strategy in the 
		bisimulation game $\B(\G,\G',Tot)$ where $Tot$ is the total relation $Hist(\G) \times Hist(\G')$.
	\end{enumerate}
\end{proposition}

The second claim follows by observing that if the \duplic{} coalition
has a strategy to win a bisimulation game $\B(R) =
\B(\G,\G',h_0,h_0',R)$ for some $R$, then they also have a
strategy to win the bisimulation game $\B(Tot) =
\B(\G,\G',h_0,h_0',Tot)$, and the construction in Theorem \ref{game}
can be used to show that the bisimulation associated with $\B(R)$ is
included in the bisimulation associated with $\B(Tot)$. Hence, this
latter is maximal w.r.t.~all bisimulations that share the strategy
simulator constructed as in Theorem \ref{game}.

\section{The Hennessy-Milner Property} \label{HMproperty}

\begin{figure}[t]\centering
	\begin{tabular}{ccc}
		$\G_1$ & $\G_2$ \\[5pt]
		\def\stpd{3.5}
		\begin{tikzpicture}[>=latex, transform shape, scale = 0.65]
		\node[initstate] (m0) at (-1.5, -0.5*\stpd) {$q_1$};
		\node[state] (m1) at (0.0, -0.5*\stpd) {$q_2$};
		\node[state] (m2) at (1.5, -0.5*\stpd) {$q_3$};
		\node[state] (m3) at (3.0, -0.5*\stpd) {$q_4$};
		
		\node[state, label=-90:{$p$}] (fin) at (1, -1.5*\stpd) {$q_\top$};
		\node[state, label=-90:{}] (sink) at (-1.5, -1.5*\stpd) {$q_\bot$};
		
		\path [-,style=dotted,shorten >=1pt, auto, node distance=7cm, semithick]
		
		(m0) edge node {1} (m1)
		(m1) edge node {2} (m2)
		(m2) edge node {1} (m3)
		;
		
		\path [->,style=solid,shorten >=1pt, auto, node distance=7cm, semithick]
		
		(m0) edge [bend right] node[left,text width = 0.8cm] {$(a,x)$ $(b,y)$} (fin)
		(m1) edge [bend right] node[midway,right,text width = 0.8cm] {$(a,x)$ $(b,y)$} (fin)
		(m2) edge [bend left] node[midway,left,text width = 0.8cm] {$(a,x)$ $(b,y)$} (fin)
		(m3) edge [bend left] node[right,text width = 0.8cm] {$(a,x)$ $(b,y)$} (fin)
		
		(sink) edge [loop left] (sink)
		
		;
		\end{tikzpicture}
		&
		\begin{tikzpicture}[>=latex, transform shape, scale = 0.65]
		\def\stpd{3.5}
		\node[initstate] (m0) at (-1.5, -0.5*\stpd) {$q'_1$};
		\node[state] (m1) at (0.0, -0.5*\stpd) {$q'_2$};
		\node[state] (m2) at (1.5, -0.5*\stpd) {$q'_3$};
		\node[state] (m3) at (3.0, -0.5*\stpd) {$q'_4$};
		
		\node[state, label=-90:{$p$}] (fin) at (1, -1.5*\stpd) {$q'_\top$};
		\node[state, label=-90:{}] (sink) at (-1.5, -1.5*\stpd) {$q'_\bot$};
		
		\path [-,style=dotted,shorten >=1pt, auto, node distance=7cm, semithick]
		
		(m0) edge node {1} (m1)
		(m1) edge node {2} (m2)
		(m2) edge node {1} (m3)
		;
		
		\path [->,style=solid,shorten >=1pt, auto, node distance=7cm, semithick]
		
		(m0) edge [bend right] node[left] {$(a,x)$} (fin)
		(m1) edge [bend right] node[midway,right,text width = 0.8cm] {$(a,x)$ $(b,y)$} (fin)
		(m2) edge [bend left] node[midway,left,text width = 0.8cm] {$(a,x)$ $(b,y)$} (fin)
		(m3) edge [bend left] node[right] {$(b,y)$} (fin)
		
		(sink) edge [loop left] (sink)
		
		;
		
		\end{tikzpicture}
		\vspace*{-10pt}
	\end{tabular}
	\caption{Counterexample for the Hennessy-Milner property for the subjective and objective semantics.
		\vspace*{-10pt}}\label{hmp}
\end{figure}
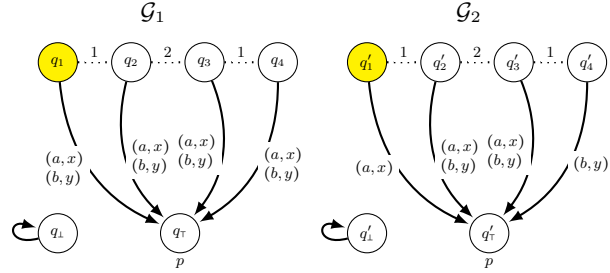

We now show that the notion of bisimulation introduced in
Sec.~\ref{altbisim} enjoys the Hennessy-Milner property. To this end,
we need to define $A$-equivalence between iCGS.
Specifically, given iCGS $\G$ and $\G'$ with histories $h_0 \in
Hist(\G)$ and $h_0' \in Hist(\G')$, we say that the pointed iCGS
$\G(h_0)$ and $\G'(h_0')$, having $h_0$ and $h'_0$ as respective
initial histories, are {\em $A$-equivalent} iff for every $A$-formula
$\phi$, $(\G, h_0) \models_x \phi$ iff $(\G', h_0') \models_x \phi$.
\begin{theorem}\label{HM-theorem}
	The notion of bisimulation in Def.~\ref{def:bisim} enjoys the
	Hennessy-Milner property, that is, the pointed iCGS $\G(h_0)$ and
	$\G'(h_0')$ are $A$-equivalent for the common knowledge semantics if and only if they are $A$-bisimilar.
\end{theorem}

Before proving Theorem \ref{HM-theorem}, as a counterexample for the subjective and objective semantics, 
we recall the example used in \cite{BCCJK20} depicted in Figure \ref{hmp}. 
In each state agent $1$ can execute actions $\{a, b, c\}$ while agent $2$ can execute $\{x, y, z\}$. 
The transitions shown lead to $q_\top$ and $q'_\top$, while the omitted transitions lead to $q_\bot$ and $q'_\bot$, respectively. 
%
We can check that states $q_i$ and $q'_j$, with $i,j\in\{1,2,3,4\}$, are $\{1,2\}$-equivalent, 
and it holds the same also for states $q_\bot$ (resp., $q_\top$) and $q'_\bot$ (resp., $q'_\top$).  
Therefore, $\G_1$ and $\G_2$ are $\{1,2\}$-equivalent, i.e., they satisfy the same $\{1,2\}$-formulas in ATL$^*$. 
However, there is no $\{1,2\}$-bisimulation between the two iCGS. In particular, for any $i, j \in \{1,2,3,4\}$, 
state $q_i$ cannot be $\{1,2\}$-bisimilar with any state $q'_j$. 

\cata{
	The lack of a bisimulation between the two structures in Fig.~\ref{hmp} follows
	by observing that the \spoiler{} coalition wins the appropriate bisimulation game,
	since in the initial position $(q_1,q_1')$, \ispoil{} may choose strategy $\sigma$ which 
	produces action tuple $(a,x)$ in each state in $C_{\{1,2\}}(q_1) = \{q_1,q_2,q_3,q_4\}$.
	If \idupl{} responds with strategy $\sigma'(q_1') = \sigma'(q_2')=\sigma'(q_3') =\sigma'(q_4)= (a,x)$,
	then \pspoil{} will choose state $q_4'$, 
	and \pdupl{} has no good choice of some state in $\G_1$ whose successor by $\sigma$ is labeled $\neg p$, which 
	is the label of the successor of $q_4'$ by $\sigma'$.
	Similar situations occur for all the other choices by \idupl{}.
}

To prove Theorem~\ref{HM-theorem}, we actually prove the following
Gale-Stewart-type theorem for the bisimulation games introduced in
Def.~\ref{def:3bisim-games}:
\begin{theorem}[Gale-Stewart theorem for bisimulation games] \label{GS-theorem}
	Each bisimulation game is \textbf{determined}: either the \duplic{} coalition or the \spoiler{} coalition wins the game.
\end{theorem}
\begin{proof}
	We follow the pattern of Gale-Stewart games by proving first that, at
	positions where one coalition does not have a winning strategy, the
	other one has a ``defensive'' strategy, and then showing that any
	defensive strategy is winning.
	
	
	Formally, a \textbf{defensive} strategy for the \spoiler{} coalition
	in the bisimulation game is a joint strategy $\sigma_S =
	(\sigma_{IS},\sigma_{PS})$ such that, for any position $p$ of the game
	which is compatible with $\sigma_D$, \duplic{} does not have a winning
	strategy \emph{starting from the set of positions that have the same
		observability as $p$ and are compatible with $\sigma_S$}.  Defensive
	strategies for \duplic{} are defined similarly.
	
	We then have:
	\begin{lemma}
		If \duplic{} coalition does not have a winning strategy,
		then \spoiler{} coalition has a defensive strategy.
	\end{lemma}
	The proof of this claim works similarly to the classical
	case \cite{GaleS53}, by building the defensive strategy by induction
	on the level of the position in the tree of positions of the
	bisimulation game, such that, at each position belonging to \pspoil{}
	or \ispoil, we identify one ``defensive'' action for these agent.  The
	difficulty is
	to build a uniform strategy for \ispoil,
	i.e., at two positions with identical observations
	for \ispoil, her actions are identical.
	
	The argument is similar for the base case and the inductive cases, and
	starts by assuming the following: for some position $p$ with $lab(p) =
	(\br h,\br h')$,
	the \spoiler{} coalition does not have a defensive strategy from
	$Obs(p)$, but neither \duplic{} has a winning strategy from $Obs(p)$.
	Therefore the following property, which formalizes the lack of a defensive strategy w.r.t. L-transitions 
	(i.e. steps 2.a, 3, 5, 7, 9 and 11) in the bisimulation game, holds:
	\begin{multline}
		\forall p_1 \in Obs(p) \ 
		\forall \sigma \in PStr(C_A(lab(p_1))) \\ 
		\exists \sigma'_{\sigma,p_1} \in PStr(C_A'(lab_2(p_1))) \ \forall k' \in C_A'(lab_2(p_1)) \\ 
		\exists k_{\sigma,p_1,k'}, \in C_A(lab_1(p_1)) \ \forall l' \in succ(k',\sigma'_{\sigma,p_1}) \\ 
		\exists l_{\sigma,p_1,k',l'} \in succ(k_{\sigma,p_1,k'},\sigma) \ \exists \sigma_D \!=\! (\sigma_{ID},\sigma_{PD}) \!\in\! \Sigma_{Dupl} \\
		\text{ with $\sigma_D$ winning from position } (l_{\sigma,p_1,k',l'},l')
		\label{formula-defensive}
	\end{multline}
	Note that a similar property holds w.r.t.~R-transitions.
	
	Notice that, by default, nothing excludes having some $p_1,p_2 \in
	Obs(p)$ such that, for some, $\sigma \in PStr(C_A(lab_1(p)))$,
	$\sigma'_{\sigma,p_1} \neq \sigma'_{\sigma,p_2}$.  But \idupl{} can choose the same
	$\sigma'_{\sigma,p}$ and \pdupl{} can then choose
	$l_{\sigma,p,k',l'}$ at \textbf{all positions} $p_1\in Obs(p)$,
	because choices of $k'$ for \pdupl{} 
	in Formula \ref{formula-defensive}
	are quantified over the whole $C_A'(lab_2(p_1))=C_A'(lab_2(p))$.
	In other words, because \idupl{}'s choice of $\sigma'_{\sigma,p}$ combined with 
	\pdupl{} choice of $l_{\sigma,p,k',l'}$
	are "defensive" at position $p$, they are "defensive" at any other position $p_1 \in Obs(p)$.
	This way, \idupl{}'s choice can be made \textbf{uniform} w.r.t. her observations, 
	which gives a winning strategy for the \duplic{} coalition at $p$, fact which
	contradicts the initial assumption.
	
	
	
	Formally, from any position $p_1 \in Obs(p)$, the following strategy for \duplic{} coalition is winning:
	\begin{itemize}
		\item Denote $p^1_\sigma$ the successor of $p_1$ after \ispoil{} chooses $\sigma$. 
		Then \idupl{} chooses $\sigma'_{\sigma,p}$ at $p^1_\sigma$. 
		\item Denote the resulting position $p^1_{\sigma,\sigma'_{\sigma,p}}$.
		Denote further by $p^1_{\sigma,\sigma'_{\sigma,p},k'}$ the successor of $p^1_{\sigma,\sigma'_{\sigma,p}}$ after 
		\pspoil{} has chosen $k'\!\in\! C_A'(lab_2(p)\!)$. 
		Then \pdupl{} chooses $k_{\sigma,p,k'} \!\in C_A(lab_1(p))$ at $p^1_{\sigma,\sigma'_{\sigma,p},k'}$. 
		\item Denote the resulting position $p_{\sigma,\sigma'_{\sigma,p},k_{\sigma,p,k'}}$.
		Also denote $p^1_{\sigma,\sigma'_{\sigma,p},k',k_{\sigma,p,k'},l'}$ the successor of 
		$p^1_{\sigma,\sigma'_{\sigma,p},k_{\sigma,p,k'}}$ 
		after \pspoil{} chooses 
		$l' \!\in\! succ(k'\!,\sigma'_{\sigma,p})$. 
		Then \pdupl{} chooses  $l_{\sigma,p,k',l'} \!\in\! succ(k_{\sigma,p,k'},\sigma)$ at 
		$p^1_{\sigma,\sigma'_{\sigma,p},k',k_{\sigma,p,k'},l'}$.
	\end{itemize}
	
	Formula \ref{formula-defensive} implies that from the resulting
	position, which is labeled $(l_{\sigma,p,k',l'},\allowbreak l')$, the
	\duplic{} coalition has a winning strategy. Hence, we have a winning strategy for the \duplic{} coalition from $Obs(p)$,
	which contradicts the initial assumption.
	
	As a result, the \spoiler{} coalition must have a defensive strategy
	from $Obs(p)$, which can be built by negating
	Formula \ref{formula-defensive}, after
	which the construction continues by induction on the observation class
	of the resulting positions labeled $(l_{\sigma,p,k',l'},l')$. A
	similar argument shows that, when \spoiler{} coalition does not have a
	winning strategy, \duplic{} coalition has a defensive strategy.
	
	It remains to show that a defensive strategy for \duplic{} is winning.
	This follows by observing that any infinite path in $T(\B)$ which is compatible with a defensive strategy
	for \duplic{} must not pass through a position which is winning
	for \spoiler{}, hence is an infinite path which is winning
	for \duplic{}, which ends the proof.
\end{proof}

We can now proceed with the proof of Theorem \ref{HM-theorem}.
\begin{proof}{\bf Theorem \ref{HM-theorem}}
	Assume that there exists no bisimulation between $\G$ and $\G'$ 
	which, by Proposition \ref{prop:tot}, means that
	in  the bisimulation
	game $\B(\G,\G',s_0,s_0',Tot)$ 
	the \duplic{} coalition has no winning
	strategy.  By the determinacy theorem, the \spoiler{} coalition has a
	winning strategy $\sigma_S = (\sigma_{IS},\sigma_{PS})$.  Since each
	position in the bisimulation game has a finite number of successors,
	as a consequence of K\"onig's Lemma, there exists a finite set
	$P_{\sigma_S}$ of winning positions for \spoiler{} such that all runs
	compatible with $\sigma_S$ pass through one position of
	$P_{\sigma_S}$.
	
	Pick then a position $p$ labeled $(h,h')$ such that, on all runs
	starting from $p$ and compatible with $\sigma_S$, the first position
	labeled with some $(l,l')$ occurring on the run after $p$ is a winning
	position for \spoiler.  Note that the following
	property, formalizing the fact that $\sigma_S$ is winning, holds:
	\begin{multline}\!\!\!\!\!\!
		\exists \sigma\! \in\! PStr(C_A(h)\!)\:
		\forall \sigma' \!\!\!\in \! PStr(C_A'(h')\!) \:
		\exists k'_{\sigma'} \!\in\! C_A'(h') \: \forall k \!\in\! C_A(h)\\ 
		\big( \pi(k) = \pi'(k'_{\sigma'}) \rightarrow \exists l'_{\sigma',k'_{\sigma'},k}\! \in\! succ(k',\sigma') \\
		\forall l \!\in \!succ(k,\sigma) \ \pi(l) \!\neq\! \pi(l'_{\sigma',k'_{\sigma'},k})\big)
		\label{formula-xatl}
	\end{multline}
	where $\sigma \!=\! \sigma_{IS}(p)$, $k'_{\sigma'} \!=\! \sigma_{PS}(p_1)$,
	$p_1$ is the successor of $p$ after \ispoil{} chooses $\sigma$
	and \pdupl{} answers with $\sigma'$, and $l'_{\sigma',k'_{\sigma'},k} \!
	= \! \sigma_{PS}(p_2)$ where $p_2$ is the successor of $p_1$
	after \pspoil{} chooses $k'$ and \pdupl{} answers with $k$.  Note that
	the implication with premise $\pi(k) = \pi'(k'_{\sigma'})$ is
	needed since \pdupl{}'s choices with $\pi(k) \!\neq\! \pi'(k'_{\sigma'})$
	are immediately winning for \spoiler{}, and then there is no need to
	proceed with steps 11-12 corresponding with the successors of $k$ and
	$k'$.
	
	So, if we define the formula 
	\begin{equation}\label{formula-atl-bisim}
		\varphi(P_{\sigma_S}) \!=\! \dlangle A \drangle X \Big(\!\!\!\!\!\!\!\!\bigwedge_{\sigma'\in PStr(C_A'(h'))}\!\!\!\!\!\!\!\!\!\!\!\!\!\big( \textbf{Y} \pi'(k'_{\sigma'}) 
		\rightarrow \!\!\!\!\!\bigvee_{k\in C_A(h)} \!\!\!\!\!\!
		\neg \pi'(l'_{\sigma',k'_{\sigma'},k})\big)\Big)
	\end{equation}
	then Formula \ref{formula-xatl} implies that $(\G, h) \models_{ck} \varphi(P_{\sigma_S})$ but, on the other hand, $(\G', h') \not \models_{ck} \varphi(P_{\sigma_S})$.
	
	\new{
		To see this, note that, in 
		$\varphi(P_{\sigma_S})$ 
		the coalition operator $\dlangle A \drangle$ encodes
		$\exists \sigma \in PStr(C_A(h))$ in \ref{formula-xatl}.  Further, the
		conjunction indexed by $\sigma'$ in \ref{formula-atl-bisim}
		corresponds to the universal quantifier on $\sigma'$
		in \ref{formula-xatl}.  The $k'_{\sigma'}$ in \ref{formula-atl-bisim}
		represents
		the skolemization of 
		the existential quantification over $k'_{\sigma'}$  in \ref{formula-xatl}. 
		The last disjunction in \ref{formula-atl-bisim} corresponds with the existential quantification over $k$ in \ref{formula-xatl}. 
		The existential quantification over $l'_{\sigma',k'_{\sigma'},k}$ in \ref{formula-xatl} is encoded in \ref{formula-atl-bisim} 
		by its skolemization, denoted $l'_{\sigma',k'_{\sigma'},k}$ too. 
		Finally, 
		the universal quantifier over $l \in succ(k,\sigma) $ in \ref{formula-xatl} and the last 
		property connecting $l$ to $l'_{\sigma',k'_{\sigma'},k}$
		is encoded in \ref{formula-atl-bisim} by $\neg \pi'(l'_{\sigma',k'_{\sigma'},k})$.
	}
	
	\new{
		The yesterday operator $Y$ is needed because
		we must encode the part of \ref{formula-xatl} referring to $\pi(k)$, which refers to the current position. 
		Unfortunately, a formula like $\dlangle A \drangle (\pi(k) \to X \psi)$ which would  simulate 
		more easily the implication $\pi(k) = \pi'(k'_{\sigma'}) \rightarrow  \exists l'_{\sigma',k'_{\sigma'},k}\! \in\! succ(k',\sigma') \ldots $ from \ref{formula-xatl}
		would not be  ATL but rather ATL$^*$.
		But, in order to correctly simulate  quantifier order from \ref{formula-xatl}, in \ref{formula-atl-bisim} $\pi(k)$ 
		must lie within the scope of $\dlangle A \drangle X$, which refers to the positions one time step after the current position. 
		Hence, in the scope of $\dlangle A \drangle X$ we need to recover the value of $\pi(k)$ at the previous position, 
		hence we utilize $Y$. We believe $Y$ might not be needed for the full $ATL^*$, a topic for further research. 
	}
	
	
	The proof can then be completed by induction as follows: we modify the
	bisimulation game by appending a new winning condition for \spoiler{}:
	all positions in $Obs(p)$ are labeled as winning, with the formula
	$\varphi$ witnessing this.  The set of atomic propositions for both
	iCGS is augmented with $p_\varphi$ and, for each $(h,h')$ labeling a
	position $p_1 \in Obs(p)$, we augment $\pi(h)$ with $p_\varphi$ while
	$\pi'(h')$ is left unchanged.  This provides us with a new
	bisimulation game in which (the appropriately updated) strategy
	profile $\sigma_S$ is still a winning strategy for \spoiler{},
	there is a strictly smaller set of positions $P'_{\sigma_S}$ which are
	winning for \spoiler{}, and all runs compatible with $\sigma_S$ pass
	through one position of $P_{\sigma_S}'$.
	
	The argument ends when we obtain some $P^m_{\sigma_S}$ for which
	$Obs(P^m_{\sigma_S})$ is a singleton, which means that $(h_0,h_0') \in
	P^m_{\sigma_S}$. Then the formula $\varphi(P^m_{\sigma_S})$ built as
	in Equation~\ref{formula-atl-bisim} is the witness that $(\G,h_0)$ is
	not $A$-equivalent with $(\G',h_0')$.
\end{proof}





\section{Undecidability Result} \label{bisim_undec}

In this section we show that deciding the existence of bisimulations
between iCGS is undecidable in general.  We state immediately the main
result of this section.
\begin{theorem} \label{undec}
	The problem of checking whether two CGS $\G_1$ and $\G_2$ defined on
	the same set $Ag$ of agents are $A$-bisimilar, for some set
	$A\subseteq Ag$ of agents, is undecidable.
\end{theorem}

The proof of Theorem~\ref{undec} can be outlined as follows: given any
deterministic Turing Machine, by building on \cite{DimaTiplea11} we
construct a 3-agent iCGS which has the property that two agents (call
them $1$ and $2$) have a winning strategy to avoid an error state if
and only if the TM never stops when starting with an empty tape.  We
note that this strategy, when it exists, is unique.  Then, we
construct a second 3-agent ``simple'' iCGS in which there exists a
unique strategy for agents $1$ and $2$ (without memory) for avoiding
an error state.  Finally, we prove that these two iCGS are bisimilar
if and only if the TM never stops when starting with an empty tape,
which is sufficient to derive the undecidability of the former
problem.

We start with the construction of the second iCGS, which is depicted in
Figure~\ref{simplemodel}.  Note that the transitions only represent
the actions of agents $1$ and $2$, agent $3$'s role is to "solve
nondeterminism" in states $s_{init}$, $s_{gen}$ and $s_{tr}$.
First of all, we prove the following lemma.
\begin{lemma} \label{remark-deviations-g-simple}
	In the iCGS depicted in Figure~\ref{simplemodel} there exists a unique
	strategy for agents $1$ and $2$ to avoid state  $s_{err}$.
\end{lemma}
\begin{proof}
	
	Note that, in all the states except $s^1_{amb}$ and $s^2_{amb}$,
	coalition $\{1,2\}$ must play $(ok,ok)$ to avoid $s_{err}$. To further
	understand why $1$ and $2$ need to play the same action in the
	remaining two states, consider history $h = s_{init}
	\xrightarrow{ok,ok} s_{gen} \xrightarrow{ok,ok} s_{tr}
	\xrightarrow{ok,ok} s^2_{amb}$.  Note that if we have $h \sim_1 h'$
	(and $h \neq h'$), then $h' = s_{init} \xrightarrow{ok,ok} s_{gen}
	\xrightarrow{ok,ok} s^1_{amb} \xrightarrow{a_1,a_2} s^1_{namb}$.  So,
	for any joint strategy $\sigma = (\sigma_1,\sigma_2)$ for $1$ and $2$,
	ensuring $s_{err} \notin succ(h',\sigma_1(h'))$ requires that
	$\sigma_1(h') = ok$, which also implies that $\sigma_1(h) = ok$ by
	$1$-uniformity of $\sigma$.  A similar argument shows that $h \sim_2
	h''$ (and $h \neq h''$) implies that $h'' = s_{init} \xrightarrow{ok,ok} s_{gen}
	\xrightarrow{ok,ok} s_{tr} \xrightarrow{ok,ok} s_{gen}$, and since the
	only way to ensure $s_{err} \notin succ(h'',\sigma_2(h''))$
	is by
	choosing $\sigma_2(h'') = ok$, we must also have $\sigma_2(h) = ok$.
	
	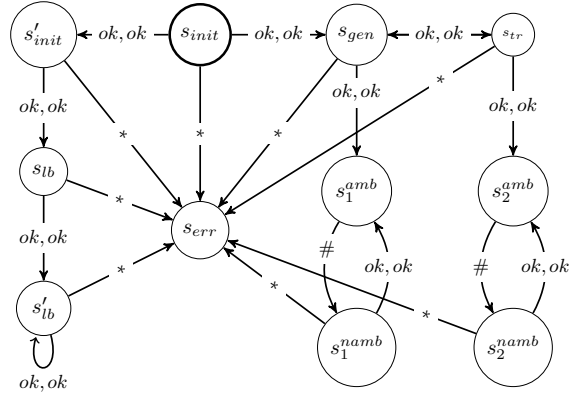
\begin{figure}
		\begin{center}
			\mbox{\scalebox{0.8}[0.8]{
					\begin{tikzpicture}
					[node distance = 4em]
					\node [cnode ]
					(Sgen)
					{\Large $\stackrel{s_{gen}}{\tiny}$};
					\node [cnode ]
					(S1amb)
					[below of = Sgen, node distance = 8em]
					{\Large $\stackrel{s^{amb}_1}{}$};
					\node [cnode, very thick]
					(Sinit)
					[left of = Sgen, node distance = 8em]
					{\Large $\stackrel{s_{init}}{}$};
					\node [cnode]
					(Sinit')
					[left of = Sinit, node distance = 8em]
					{\Large $\stackrel{s'_{init}}{}$};
					\node [cnode]
					(Slb)
					[below of = Sinit', node distance = 7em]
					{\Large $\stackrel{s_{lb}}{}$};
					\node [cnode]
					(Slb')
					[below of = Slb, node distance = 7em]
					{\Large $\stackrel{s'_{lb}}{}$};
					\node [cnode ]
					(Str)
					[right of = Sgen, node distance = 8em]
					{ $\stackrel{s_{tr}}{\tiny}$};
					\node [cnode ]
					(S2amb)
					[below of = Str, node distance = 8em]
					{\Large $\stackrel{s^{amb}_2}{}$};
					\node [cnode ]
					(S2namb)
					[below of = S2amb, node distance = 8em]
					{\Large $\stackrel{s^{namb}_2}{}$};
					\node [cnode ]
					(S1namb)
					[below of = S1amb, node distance = 8em]
					{\Large $\stackrel{s^{namb}_1}{}$};
					\node [cnode]
					(Serr)
					[below of = Sinit, node distance = 10em]
					{\Large $\stackrel{s_{err}}{}$};
					
					\path[-stealth']
					(Sgen)
					edge  [pos = 0.5]
					node [] {\footnotesize $ok,ok$}
					(Str)
					edge	[pos = 0.5]
					node [] {\footnotesize $*$}
					(Serr)
					edge  [pos = 0.3]
					node [] {\footnotesize $ok,ok$}
					(S1amb)
					
					(Sinit)
					edge	[pos = 0.5]
					node [] {\footnotesize $*$}
					(Serr)
					edge  [pos = 0.5]
					node [] {\footnotesize $ok,ok$}
					(Sgen)
					edge  [pos = 0.5]
					node [] {\footnotesize $ok,ok$}
					(Sinit')
					
					(Sinit')
					edge	[pos = 0.5]
					node [] {\footnotesize $*$}
					(Serr)
					edge  [pos = 0.5]
					node [] {\footnotesize $ok,ok$}
					(Slb)
					
					(Slb)
					edge	[pos = 0.5]
					node [] {\footnotesize $*$}
					(Serr)
					edge  [pos = 0.5]
					node [] {\footnotesize $ok,ok$}
					(Slb')
					
					(Slb')
					edge	[pos = 0.5]
					node [] {\footnotesize $*$}
					(Serr)
					edge  [loop below]
					node [noall] {\footnotesize $ok,ok$}
					()
					
					(S2amb)
					edge	[pos = 0.5, bend right]
					node [] {\footnotesize $\#$}
					(S2namb)
					
					(Str)	edge	[pos = 0.5]
					node [] {\footnotesize $ok,ok$}
					(Sgen)
					edge	[pos = 0.2
					]
					node [] {\footnotesize $*$}
					(Serr)
					edge	[pos = 0.5]
					node [] {\footnotesize $ok,ok$}
					(S2amb)
					
					(S1amb)	
					edge	[pos = 0.3, bend right]
					node [] {\footnotesize $\#$}
					(S1namb)
					
					(S1namb)
					edge	[pos = 0.5]
					node [] {\footnotesize $*$}
					(Serr)	
					edge	[pos = 0.5, bend right]
					node [] {\footnotesize $ok,ok$}
					(S1amb)
					
					(S2namb)
					edge	[pos = 0.2
					]
					node [] {\footnotesize $*$}
					(Serr)	
					edge	[pos = 0.5, bend right]
					node [] {\footnotesize $ok,ok$}
					(S2amb)
					
					
					;
					
					
					\end{tikzpicture}
				}
			}
			\caption{\label{simplemodel}The $iCGS$ $\gs$, where $* = \neg ok,ok \mid ok, \neg ok \mid \neg ok, \neg ok$ and $\# = \neg ok,ok \mid ok, \neg ok \mid \neg ok, \neg ok \mid ok, ok $. The indistinguishability relation for player $1$ has three classes: $s_{gen}$, $s_{err}$, and $s_{set}$, where $s_{set} = S \setminus \{s_{gen}, s_{err}\}$.  The indistinguishability relation for player $2$ has three classes: $s_{tr}$, $s_{err}$, and $s'_{set}$, where $s'_{set} = S \setminus \{s_{tr}, s_{err}\}$.  
				\vspace*{-10pt}}
		\end{center}
	\end{figure}
	
	By generalizing these observations, a strategy that avoids the error
	state for all outcomes can only be constructed if in every history $h$
	ending in $s_{amb}^1$ the joint action is $(ok,ok)$.  This is because:
	\begin{enumerate}
		\item if $h\sim_1 h'$ and $last(h') = s_{namb}^2$, then in $s_{namb}^2$ the only ``good'' transition for agent $1$ is $ok$, and
		\item if $h\sim_2 h''$ and $last(h'') \in \{s_{gen},s_{namb}^1\}$, then in both these states the only ``good'' transition for agent $2$ is $ok$ too.
	\end{enumerate}
	A similar remark holds for histories ending in the other ``ambiguous''
	state, $s_{amb}^2$.  So there is only one joint uniform strategy for
	$\{1,2\}$ that avoids $s_{err}$, which is choosing $(ok,ok)$ at every
	history.	
\end{proof}

We now turn to the construction of the first iCGS, which is adapted
from \cite{DimaTiplea11}.
We give the construction for a
simple deterministic $TM$ $M = \langle \mathcal{Q}, \Gamma, \delta \rangle$ with states $\mathcal{Q} = \{ q_0, q_1, q_2\}$ and tape symbols $\Gamma = \{ B, a\}$, whose transition function $\delta$ is in Table \ref{trans-tm}. 

\begin{table}[ht]\
	\vspace*{-10pt}\\
	\centerline{
		\begin{tabular}{|c||c|c|}
			\hline
			$\delta$ & $B$ & $a$
			\\ \hline\hline
			$q_0$  & $(q_1,a,R)$ & $(q_1,B,R)$
			\\ \hline
			$q_1$  & $(q_2,a,R)$ & $(q_0,B,L)$
			\\ \hline
			$q_2$  & $(q_0,B,L)$ & $(q_2,a,R)$
			\\ \hline
		\end{tabular}
	}
	\caption{Transitions of the $TM$ $M$.\vspace*{-20pt}}
	\label{trans-tm}
\end{table}

The purpose of this construction is
that the given TM never halts on an empty tape if and only if 
coalition $\{1,2\}$ has a strategy which simulates the run of the TM
on the \emph{levels} of the tree of runs compatible with the strategy. 
The simulation of the Turing machine, depicted in
Fig.~\ref{simtm}, satisfies the following properties:

\begin{figure*}
	\begin{center}
		\mbox{\scalebox{0.59}[0.59]{
				\begin{tikzpicture}
				[node distance = 4em]
				\node [cnode,  ]
				(Sgen)
				{\Large $\stackrel{s_{gen}}{\tiny}$};
				\node [cnode,  ]
				(S1amb)
				[below of = Sgen, node distance = 6em]
				{\Large $\stackrel{s^0_B}{}$};
				\node [cnode,  ]
				(Sq0b)
				[below left of = S1amb, node distance = 10em]
				{\Large $\stackrel{s_{q_0B}}{}$};
				\node [cnode,  ]
				(Sq1b)
				[below of = S1amb, node distance = 10em]
				{\Large $\stackrel{s_{q_1B}}{}$};
				\node [cnode,  ]
				(Sq2b)
				[below right of = S1amb, node distance = 10em]
				{\Large $\stackrel{s_{q_2B}}{}$};
				\node [cnode, very thick]
				(Sinit)
				[left of = Sgen, node distance = 10em]
				{\Large $\stackrel{s_{init}}{}$};
				\node [cnode]
				(Sinit')
				[left of = Sinit, node distance = 10em]
				{\Large $\stackrel{s'_{init}}{}$};
				\node [cnode]
				(Slb)
				[below of = Sinit', node distance = 10em]
				{\Large $\stackrel{s_{lb}}{}$};
				\node [cnode]
				(Slb')
				[below of = Slb, node distance = 10em]
				{\Large $\stackrel{s'_{lb}}{}$};
				\node [cnode,  ]
				(Str)
				[right of = Sgen, node distance = 30em]
				{ $\stackrel{s_{tr}}{\tiny}$};
				\node [cnode,  ]
				(S2amb)
				[below of = Str, node distance = 10em]
				{\Large $\stackrel{s^{1}_{tr}}{}$};
				\node [cnode,  ]
				(Sq0q1r)
				[above left of = S2amb, node distance = 10em]
				{\Large $\stackrel{s_{q_0q_1R}}{}$};
				\node [cnode,  ]
				(Sq2q0l)
				[left of = S2amb, node distance = 10em]
				{\Large $\stackrel{s_{q_2q_0L}}{}$};
				\node [cnode,  ]
				(Sq2q2r)
				[below left of = S2amb, node distance = 10em]
				{\Large $\stackrel{s_{q_2q_2R}}{}$};
				\node [cnode,  ]
				(S2namb)
				[above right of = S2amb, node distance = 10em]
				{\Large $\stackrel{s^{0}_{tr}}{}$};
				\node [cnode,  ]
				(Sq1q0l)
				[right of = S2amb, node distance = 10em]
				{\Large $\stackrel{s_{q_1q_0L}}{}$};
				\node [cnode,  ]
				(Sq1q2r)
				[below right of = S2amb, node distance = 10em]
				{\Large $\stackrel{s_{q_1q_2R}}{}$};
				\node [cnode,  ]
				(S1namb)
				[below of = Sinit, node distance = 4em]
				{\Large $\stackrel{s^{1}_B}{}$};
				\node [cnode,  ]
				(Sa0)
				[below of = Sq1b, node distance = 5em]
				{\Large $\stackrel{s^0_a}{}$};
				\node [cnode,  ]
				(Sa1)
				[right of = Sa0, node distance = 10em]
				{\Large $\stackrel{s^{1}_a}{}$};
				\node [cnode,  ]
				(Sq0a)
				[below left of = Sa0, node distance = 10em]
				{\Large $\stackrel{s_{q_0a}}{}$};
				\node [cnode,  ]
				(Sq1a)
				[left of = Sa0, node distance = 9em]
				{\Large $\stackrel{s_{q_1a}}{}$};
				\node [cnode,  ]
				(Sq2a)
				[below right of = Sa0, node distance = 10em]
				{\Large $\stackrel{s_{q_2a}}{}$};
				\node [cnode,  ]
				(Sw1)
				[right of = S1amb, node distance = 7em]
				{\Large $\stackrel{s_{\bot}^1}{}$};
				\node [cnode,  ]
				(Sw2)
				[right of = Sw1, node distance = 8em]
				{\Large $\stackrel{s_{\bot}^2}{}$};
				\node [cnode,  ]
				(Sw3)
				[below of = S2amb, node distance = 10em]
				{\Large $\stackrel{s_{\bot}^3}{}$};
				\node [cnode,  ]
				(Sw4)
				[below of = Sw3, node distance = 7em]
				{\Large $\stackrel{s_{\bot}^4}{}$};
				
				\path[-stealth']
				(Sw1)
				edge  [pos = 0.5, bend left]
				node [] {\footnotesize $i,i$}
				(Sw2)
				
				(Sw2)
				edge  [pos = 0.5, bend left]
				node [] {\footnotesize $any$}
				(Sw1)
				
				(Sw3)
				edge  [pos = 0.5, bend left]
				node [] {\footnotesize $i,i$}
				(Sw4)
				
				(Sw4)
				edge  [pos = 0.5, bend left]
				node [] {\footnotesize $any$}
				(Sw3)
				
				(Sgen)
				edge  [pos = 0.5]
				node [] {\footnotesize $i,i$}
				(Str)
				edge  [pos = 0.5]
				node [] {\footnotesize $i,i$}
				(S1amb)
				
				(Sinit)
				edge  [pos = 0.5]
				node [] {\footnotesize $i,i$}
				(Sgen)
				edge  [pos = 0.5]
				node [] {\footnotesize $i,i$}
				(Sinit')
				
				(Sinit')
				edge  [pos = 0.5]
				node [] {\footnotesize $i,i$}
				(Slb)
				
				(Slb)
				edge  [pos = 0.5]
				node [] {\footnotesize $i,q_0$}
				(Slb')
				
				(Slb')
				edge  [loop below]
				node [noall] {\footnotesize $i,i$}
				()
				
				(S2amb)
				edge	[pos = 0.5]
				node [] {\footnotesize $i,i$}
				(S2namb)
				edge	[pos = 0.5]
				node [] {\footnotesize $(q_0q_1R),i$}
				(Sq0q1r)
				edge	[pos = 0.5, bend right]
				node [] {\footnotesize $i,(q_2q_0L)$}
				(Sq2q0l)
				edge	[pos = 0.5]
				node [] {\footnotesize $(q_2q_2R),i$}
				(Sq2q2r)
				edge	[pos = 0.5]
				node [] {\footnotesize $(q_1q_2R),i$}
				(Sq1q2r)
				edge	[pos = 0.5, bend left]
				node [] {\footnotesize $i,(q_1q_0L)$}
				(Sq1q0l)
				edge	[pos = 0.8]
				node [] {\footnotesize $+$}
				(Sw3)
				
				(Sq0q1r)
				edge	[pos = 0.5, bend left]
				node [] {\footnotesize $i, (q_0q_1R)$}
				(S2amb)
				
				(Sq2q0l)
				edge	[pos = 0.5]
				node [] {\footnotesize $(q_2q_0L),i$}
				(S2amb)
				
				(Sq2q2r)					
				edge	[pos = 0.5, bend right]
				node [] {\footnotesize $i, (q_2q_2R)$}
				(S2amb)					
				
				(Sq1q2r)
				edge	[pos = 0.5, bend left]
				node [] {\footnotesize $i, (q_1q_2R)$}
				(S2amb)	
				
				(Sq1q0l)			
				edge	[pos = 0.5]
				node [] {\footnotesize $(q_1q_0L), i$}
				(S2amb)					
				
				(Str)	edge	[pos = 0.5]
				node [] {\footnotesize $i,i$}
				(Sgen)
				edge	[pos = 0.3]
				node [] {\footnotesize $i,i$}
				(S2amb)
				
				(S1amb)	
				edge	[pos = 0.5, bend right = 20]
				node [] {\footnotesize $i,i$}
				(S1namb)
				edge	[pos = 0.5]
				node [] {\footnotesize $*$}
				(Sq0b)
				edge	[pos = 0.5]
				node [] {\footnotesize $i,(q_0q_1R)$}
				(Sq1b)
				edge	[pos = 0.5]
				node [] {\footnotesize $\#$}
				(Sq2b)
				edge	[pos = 0.5]
				node [] {\footnotesize $+$}
				(Sw1)
				
				(Sq1b)
				edge	[pos = 0.5]
				node [] {\footnotesize $(q_1q_2R),i$}
				(Sa0)
				
				(Sq0b)
				edge	[pos = 0.5]
				node [] {\footnotesize $(q_0q_1R),i$}
				(Sa0)	
				
				(Sq2b)
				edge	[pos = 0.5, bend right]
				node [] {\footnotesize $i, (q_2q_0L)$}
				(S1amb)					
				
				(S1namb)
				edge	[pos = 0.5, bend left = 20]
				node [] {\footnotesize $i,i$}
				(S1amb)	
				
				(Sa0)	
				edge	[pos = 0.5]
				node [] {\footnotesize $i,i$}
				(Sa1)
				edge	[pos = 0.6]
				node [] {\footnotesize $\stackrel{(q_1,q_0,L),i}{(q_2,q_0,L),i}$}
				(Sq0a)
				edge	[pos = 0.45]
				node [] {\footnotesize $i,(q_0,q_1,R)$}
				(Sq1a)
				edge	[pos = 0.5, bend left]
				node [] {\footnotesize $\stackrel{i,(q_1q_2R)}{i,(q_2q_2R)}$}
				(Sq2a)
				edge	[pos = 0.5, bend right = 65]
				node [] {\footnotesize $+$}
				(Sw1)
				
				(Sq2a)
				edge	[pos = 0.3, bend left]
				node [] {\footnotesize $i, (q_2,q_2,R)$}
				(Sa0)
				
				(Sq0a)
				edge	[pos = 0.2, bend left = 90]
				node [] {\footnotesize $(q_0,q_1,R),i$}
				(S1amb)	
				
				(Sq1a)
				edge	[pos = 0.1, bend left = 60]
				node [] {\footnotesize $i, (q_1,q_0,L)$}
				(S1amb)						
				
				(Sa1)
				edge	[pos = 0.5]
				node [] {\footnotesize $i,i$}
				(Sa0)
				
				(S2namb)
				edge	[pos = 0.5]
				node [] {\footnotesize $i,i$}
				(S2amb)
				;
				\end{tikzpicture}
			}
		}
		\caption{\label{tm-model}The $iCGS$ $\gtm$, where $* = i, q_0 \mid (q_1,q_0,L), i \mid (q_2,q_0,L), i$, $\# = i, (q_1,q_2,R) \mid i, (q_2,q_2,R)$, $+$ represents all the possible combination of actions less the tuples already displayed, and $any$ represents all the possible combinations of actions. Note that, all the missing transitions go to the error state. The indistinguishability relation for player $1$ has three classes: $s_{gen}$, $s_{err}$, and $s_{set}$, where $s_{set} = S \setminus \{s_{gen}, s_{err}\}$.  The indistinguishability relation for player $2$ has three classes: $s_{tr}$, $s_{err}$, and $s'_{set}$, where $s'_{set} = S \setminus \{s_{tr}, s_{err}\}$.
			\vspace*{-15pt}}
	\end{center}
\end{figure*}

\begin{enumerate}
	\item Every run $\rho$ starting with $s_{init}\big(\xrightarrow{i,i}
	s_{gen} \xrightarrow{i,i} s_{tr}\big)^n \xrightarrow{i,i} s_{gen}
	\xrightarrow{i,i} s_B$, simulates the evolution of the contents of
	the $n$-th cell on the tape. We call such runs as {\em $(1,n)$-runs}
	and denote them $\rho^{(1,n)}$.  Formally, for each $k\geq n$,
	depending on the $k$-th configuration of the TM:
	\begin{enumerate}
		\item If the R/W head points to cell $n$ that holds symbol $x$,
		the TM is in state $q$, and the transition table gives
		$\delta(q,x) \!=\! (r,y,R)$ (i.e. an R-move of the head), then \\
		$\rho^{(1,n)}[2k+2,2k+4] = s_{q,x}\xrightarrow{(q,r,R),i} s_y^0 \xrightarrow{i,i} s_y^1$\\
		$\rho^{(1,n+1)}[2k+2,2k+4] = s_z^1\xrightarrow{i,i} s_z^0 \xrightarrow{i,(q,r,R)} s_{r,z}$, for some $z \!\in\! \Gamma$
		representing the contents of tape cell $(n\!+\!1)$ in configuration $k$.
		
		\item On the other hand,
		if the transition table gives $\delta(q,x) \\ \!=\! (r,y,L)$ (i.e. an L-move of the head), 
		then \\
		$\rho^{(1,n)}[2k+2,2k+4]  = s_{q,x}\xrightarrow{i,(q,r,L)} s_y^0 \xrightarrow{i,i} s_y^1$\\
		$\rho^{(1,n-1)}[2k+2,2k+4] = s_z^1\xrightarrow{i,i} s_z^0 \xrightarrow{(q,r,L),i} s_{r,z}$, for some $z \in \Gamma$.
		
		\item Otherwise (i.e., the R/W head is not pointing cells $n-1$ to
		$n+1$), $\rho^{(1,n)}[2k,2k+2] = s_z^1 \xrightarrow{i,i} s_z^0
		\xrightarrow{i,i} s_z^1$ where $z$ is the contents of cell $n$
		in configuration $k$.
	\end{enumerate}
	Note that two steps are needed along each run to encode the transition of the R/W head from 
	cell $n$ to cell $(n+1)$ for an R-move, or to cell $(n-1)$ for a L-move.
	
	\item Every run $\rho$ starting with $s_{init}\big(\xrightarrow{i,i}
	s_{gen} \xrightarrow{i,i} s_{tr}\big)^n \xrightarrow{i,i} s^1_{tr}$,
	simulates
	a move of the R/W head between the $(n-1)$-th and
	the $n$-th cell, (which we call in the sequel the \emph{$n$-th
		frontier}), for $n\geq 1$.  We call such runs as {\em $(2,n)$-runs} and
	denote them $\rho^{(2,n)}$.  Formally, for every $k\geq n$, depending
	on the transition between the $k$-th and the $(k+1)$-th
	configuration of the TM:
	\begin{enumerate}
		\item If the R/W head moves from the $n$-th cell to the $(n-1)$-th
		by executing $\delta(q,x) = (r,y,L)$, then\\
		$\rho^{(2,n)}[2k\!+\!3,2k\!+\!5] = s_{tr}^1
		\xrightarrow{i,(q,r,L)} s_{q,r,L} \xrightarrow{(q,r,L),i}
		s_{tr}^1$.
		\item If the R/W head moves from the $(n-1)$-th cell to the $n$-th by executing transition $\delta(q,x) = (r,y,R)$, then 
		$\rho^{(2,n)}[2k\!+\!3,2k\!+\!5] = s_{tr}^1 \xrightarrow{(q,r,R),i} s_{q,r,R} \xrightarrow{i,(q,r,R)} s_{tr}^1$.
		\item Otherwise, $\rho^{(2,n)}[2k\!+\!3,2k\!+\!5] = s_{tr}^1 \xrightarrow {i,i} s_{tr}^0\xrightarrow {i,i} s_{tr}^1$.
	\end{enumerate}
\end{enumerate}

\noindent
Finally, run $\rho_0 = s_{init} \xrightarrow{i,i} s_{init}'
\xrightarrow{i,i} s_{lb}\xrightarrow{i,(q_0)} s_{lb}'\big(
\xrightarrow{i,i} s_{lb}'\big)^\omega$ simulates the ``left bound'' of
the tape and the start of the TM with the
R/W head on the initial state $q_0$.
Let $GR$ denote the set
$\{\rho_0 \} \cup \{\rho^{(1,n)},\rho^{(2,n)}\mid n\in \mathbb{N}\}$
of ``good'' runs , and $\sigma^{win}$ the unique strategy for agents
$1$ and $2$ that simulates the infinite run of the TM, when it exists.

In what follows, we group states of $\gtm$ in five sets:
\begin{enumerate}
	\item $S^0_{namb} = \{s_{init},s_{init}',s_{lb},s_{lb}'\}$.
	\item $S^1_{namb} = \{s_{q,z}\mid  q \in \mathcal{Q}  \text{ and } z \in \Gamma\}\cup \{s_\bot^1, s_B^1, s_a^1\}$.
	\item $S^1_{amb} = \{s_a^0,s_B^0,s_\bot^2\}$.
	\item $S^2_{namb} = \{s_{tr}^1,s_\bot^3\}$.
	\item $S^2_{amb} = \big\{s_{q,r,x}\mid q,r \!\in\! Q, x \!\in \!\{R,L\}\big\}\!\cup\! \{s_\bot^4, s_{tr}^0\}$.
\end{enumerate}
We also call states in $S_{amb} = S^1_{amb}\cup S^2_{amb}$ as
``ambiguous'' and in $S_{namb} = S \setminus S_{amb}$ ``nonambiguous''. 
Note that, in each nonambiguous state, all outgoing transitions which avoir $s_{err}$ are labeled with a unique tuple of actions,
while  all ambiguous states do not have transitions leading to $s_{err}$.

Before defining the actual bisimulation between $\gtm$ and
$\gs$, note that runs that ``deviate'' from runs in $GR$ (and
therefore associated with a ``wrong'' strategy that cannot avoid
$s_{err}$) contain either a transition from a nonambiguous state to $s_{err}$,
or a transition from an ambiguous state which is not consistent with the TM computation, as explain below.

The line of reasoning guaranteeing that every strategy $\sigma$
consistent with a run which ``deviates'' from a run in $GR$ will not
be able to avoid $s_{err}$ is similar to the proof of
Lemma~\ref{remark-deviations-g-simple} above.  Assume that there
exists a single partial strategy, defined on $Hist^{\leq i}(\gtm)$,
which avoids $s_{err}$ and all histories compatible with this strategy
are prefixes of length $\leq n$ of runs in $GR$.  Note first that each
such history ending in states in $S_{namb}$ can only be completed with
a transition that simulates correctly the unique run up to level
$i\!+\!1$.

\begin{figure}[t]
	\mbox{\scalebox{0.67}[0.67]{
			\begin{tikzpicture}
			[level distance=2cm,
			level 1/.style={sibling distance=5.5cm},
			level 2/.style={sibling distance=5.5cm},
			level 3/.style={sibling distance=5cm},
			level 4/.style={sibling distance=4.5cm},
			level 5/.style={sibling distance=3cm},
			every text node part/.style={font=\small},
			every lower node part/.style={font=\small}]
			\node[ellipse,draw] {$s_{init}$}
			child {
				node[circle,draw] {$s_{init}'$}
				child {
					node[circle,draw] {$s_{lb}$}
					child {
						node[circle,draw] {$s_{lb}'$}
						child {
							node[circle,draw] {$s_{lb}'$}
							child {
								node[circle,draw] {$s_{lb}'$}
								child {
									node[circle,draw] {$s_{lb}'$}
									child {
										node[circle,draw] {$s_{lb}'$}
										child {
											node[circle,draw] {$s_{lb}'$}
											child {
												node[circle,draw] {$s_{lb}'$}
												edge from parent node[midway,left] {$(i,i)$}
											}
											edge from parent node[midway,left] {$(i,i)$}
										}
										edge from parent node[midway,left] {$(i,i)$}
									}
									edge from parent node[midway,left] {$(i,i)$}
								}
								edge from parent node[midway,left] {$(i,i)$}
							}
							edge from parent node[midway,left] {$(i,i)$}
						} 
						edge from parent node[midway,left] {$(i,(q_0))$}
					}
					edge from parent node[midway,left] {$(i,i)$}
				} 	
				edge from parent node[midway,left] {$(i,i)$}
			}
			child {
				node[ellipse,draw] {$s_{gen}$} 
				child {
					node[circle,draw] {$s^0_B$}
					child {
						node[circle,draw] {$s_{q_0,B}$}
						child {
							node[circle,draw] {$s^0_a$}
							child {
								node[circle,draw] {$s^1_a$}
								child {
									node[circle,draw] {$s^0_a$}
									child {
										node[circle,draw] {$s^1_{a}$}
										child {
											node[circle,draw] {$s^0_{a}$}
											child {
												node[circle,draw] {$s^1_{a}$}
												edge from parent node[midway,left] {$(i,i)$}
											}
											edge from parent node[midway,left] {$(i,i)$}
										}
										edge from parent node[midway,left] {$(i,i)$}
									}
									edge from parent node[midway,left] {$(i,i)$}
								}
								edge from parent node[midway,left] {$(i,i)$}
							}
							edge from parent node[midway,left] {$((q_0,q_1,R),i)$}
						}
						edge from parent node[midway,left] {$(i,(q_0))$}
					}
					edge from parent node[midway,left] {$(i,i)$}
				} 
				child {
					node[ellipse,draw] {$s_{tr}$}
					child {
						node[circle,draw] {$s^1_{tr}$}
						child {
							node[circle,draw] {$s_{q_0,q_1,R}$}
							child {
								node[circle,draw] {$s^1_{tr}$}
								child {
									node[circle,draw] {$s^0_{tr}$}
									child {
										node[circle,draw] {$s^1_{tr}$}
										child {
											node[circle,draw] {$s^0_{tr}$}
											child {
												node[circle,draw] {$s^1_{tr}$}
												edge from parent node[midway,left] {$(i,i)$}
											}
											edge from parent node[midway,left] {$(i,i)$}
										}
										edge from parent node[midway,left] {$(i,i)$}
									}
									edge from parent node[midway,left] {$(i,i)$}
								}
								edge from parent node[midway,left] {$(i,(q_0,q_1,R))$}
							}
							edge from parent node[midway,left] {$((q_0,q_1,R),i)$}
						}
						edge from parent node[midway,left] {$(i,i)$}
					}
					child { 
						node[ellipse,draw] {$s_{gen}$}
						child {
							node[circle,draw] {$s^0_B$}
							child {
								node[circle,draw] {$s_{q_1,B}$}
								child {
									node[circle,draw] {$s^0_a$}
									child {
										node[circle,draw] {$s^1_a$}
										child {
											node[circle,draw] {$s^0_a$}
											child {
												node[circle,draw] {$s_{q_0,a}$}
												edge from parent node[midway,left] {$((q_2,q_0,L),i)$}
											}
											edge from parent node[midway,left] {$(i,i)$}
										}
										edge from parent node[midway,left] {$(i,i)$}
									}
									edge from parent node[midway,left] {$((q_1,q_2,R),i)$}
								}
								edge from parent node[midway,left] {$(i,(q_0,q_1,R))$}
							}
							edge from parent node[midway,left] {$(i,i)$}
						}
						child {
							node[ellipse,draw] {$s_{tr}$}
							child {
								node[ellipse,draw] {$s^1_{tr}$}
								child {
									node[ellipse,draw] {$s_{q_1,q_2,R}$}
									child {
										node[ellipse,draw] {$s^1_{tr}$}
										child {
											node[ellipse,draw] {$s_{q_2,q_0,L}$}
											child {
												node[ellipse,draw] {$s^1_{tr}$}
												edge from parent node[midway,left] {$((q_2,q_0,L),i)$}
											}
											edge from parent node[midway,left] {$(i,(q_2,q_0,L))$}
										}
										edge from parent node[midway,left] {$(i,(q_1,q_2,R))$}
									}
									edge from parent node[midway,left] {$((q_1,q_2,R),i)$}
								}
								edge from parent node[midway,left] {$(i,i)$}
							}
							child {
								node[ellipse,draw] {$s_{gen}$}
								child {
									node[ellipse,draw] {$s^0_{B}$}
									child {
										node[ellipse,draw] {$s_{q_2,B}$}
										child {
											node[ellipse,draw] {$s^0_{B}$}
											child {
												node[ellipse,draw] {$s^1_{B}$}
												edge from parent node[midway,left] {$(i,i)$}
											}
											edge from parent node[midway,left] {$(i,(q_2,q_0,L))$}
										}
										edge from parent node[midway,left] {$(i,(q_1,q_2,R))$}
									}
									edge from parent node[midway,left] {$(i,i)$}
								}
								edge from parent node[midway,right] {$(i,i)$}
							}
							edge from parent node[midway,right] {$(i,i)$}
						}
						edge from parent node[midway,right] {$(i,i)$}
					}
					edge from parent node[midway,right] {$(i,i)$}
				}     
				edge from parent node[midway,right] {$(i,i)$}
			}
			;
			\end{tikzpicture}
	}}
	\caption{\label{simtm} Simulating three computation steps of the Turing machine in Table \ref{trans-tm}.\vspace*{-10pt}}
\end{figure}

To see what happens with the other type of histories, consider some
history $h = \rho^{(1,n)}[\!\leq \!i] \xrightarrow{a_1,a_2} s$ with $h[i]
\not \in S_{namb}$.  Note then that if $h \sim_1 h'$ then $h'[\!\leq\! i]
\preceq \rho^{(2,n)}$ 
and $h'[i]\in S_{namb}$ is a nonambiguous state.
Therefore, there exists a unique $h'[i]\xrightarrow{a_1,b_2} s'$, and,
moreover, this transition has to correctly simulate the $n$-th frontier
at level $i+1$.  It then follows that $a_1$ is the ``good'' decision
agent $1$ has to make to correctly simulate the TM run on $h$.  A
similar argument holds for $h'' \sim_2 h$ by noting that $h'' [\!\leq\! i]
\preceq \rho^{(2,n-1)}$.  Also similar arguments hold if we start with
$h \preceq \rho^{(2,n)}$.

Finally, the bisimulation relation between $\gtm$ and $\gs$
is guided by the intuition that histories ending 
in nonambiguous states $S^1_{namb}$ resp. $S^2_{namb}$, "behave similarly" with histories ending in $s_{namb}^1$, resp. $s_{namb}^2$,
while histories ending in ambiguous states $S^1_{amb}$, resp. $S^2_{namb}$ behave similarly with histories ending in $s_{amb}^1$, resp. $s_{amb}^2$.

Specifically, for every $h\in Hist(\gtm)$ with $h\prec \rho^{(1,n)}$
and $h\not \prec\rho^{(2,n)}$, we set $h \altbisim{1,2} \chi$ for
$\chi$ defined as follows:
\begin{enumerate}
	\item $\chi[i]  = h[i] $ if $h[i] \in S^0_{namb} \cup \{s_{err}\}$.
	\item $\chi[i] = s_{amb}^1$ if $h[i] \in S^1_{amb}$.
	\item $\chi[i] = s_{namb}^1$ if $h[i] \in S^1_{namb}$.
	\item For $y \in \{1,2\}$, $act_y(\chi,i) = ok$ if and only if $act_y(h,i)$ is the ``correct'' action executed by agent $y$  
	for simulating the contents of the $n$-th cell at level $i$ along $\rho^{(1,n)}$.
\end{enumerate}

Similarly, for every $h\in Hist(\gtm)$ with $h\prec \rho^{(2,n)}$ and $h\not \prec\rho^{(1,n+1)}$,
we set $h \altbisim{1,2} \chi$ for $\chi$ defined as follows:
\begin{enumerate}
	\item $\chi[i] = s_{amb}^2$ if $h[i] \in S^2_{amb}$.
	\item $\chi[i] \!=\! s_{namb}^2$ if $h[i] \!\in\! S^2_{namb}$.
	\item For $y \in \{1,2\}$, $act_y(\chi,i) = ok$ if and only if $act_y(h,i)$ is the ``correct'' action executed by agent $y$
	for simulating the $n$-th frontier at level $i$ along $\rho^{(2,n)}$.
\end{enumerate}
Note that $\altbisim{1,2}$ is in fact functional.

The strategy simulator $ST$ can be constructed using the functional relation $\altbisim{1,2}$ as follows:
for any joint strategy $\sigma$ 
and any history $h \in Hist(\gtm)$ compatible with $\sigma$,
take the unique $\chi_h\in Hist(\gs)$ with $h \altbisim{1,2}\chi_h$ and define the partial strategy $\tau_\sigma$ 
with $\tau^\sigma(\chi_h[\!<\!|\chi_h|]) = act(\chi_h,|\!<\!\chi_h|)$.
Then note that, whenever we have two different joint strategies $\sigma^1$ and $\sigma^2$ 
which share some compatible histories, then for any $h$ compatible with both we have 
that $\tau^{\sigma^1}(\chi_h[\!<\!|\chi_h|]) = \tau^{\sigma_2}(\chi_h[|\!<\!\chi_h|])$.
This means that the following definition correctly constructs a strategy simulator:
for each $h \in Hist(\gtm)$, 
each partial strategy $\br \sigma \in PStr(C_A(h))$, $h' \in C_A(h)$ and for each joint strategy  $\sigma$ with $\sigma\restr{C_A(h)} = \br \sigma$,
\begin{center}
	$ST(\br\sigma)(h) = \tau^\sigma(\chi_h[<|\chi_h|])$
\end{center}
because, as noted above, different $\tau^\sigma$ agree on the same $\chi_h$, so the choice of 
$\sigma$ is not important as long as it is compatible with $h$.

The inverse strategy simulator can be chosen as any inverse function of $ST$, i.e. any function $ST'$ with $ST\circ ST'$ being the identity function.

Hence $\altbisim{1,2}$ is indeed an $\{1,2\}$-bisimulation with $ST$ and $ST'$ strategy simulators, 
if and only if $M$ never stops when starting with an empty tape, 
which ends the proof of the undecidability theorem.

\section{Conclusions and Future Work} \label{conc}

In this paper we advanced the state of the art in the model theory of
logics for strategic reasoning in multi-agent systems.  Specifically,
in Sec.~\ref{preliminaries} we considered the {\em common knowledge}
interpretation of the Alternating-time Temporal Logic ATL under the
assumption of imperfect information (and perfect recall), which has so
far received little attention in the literature.  For this context of
imperfect information, we introduced a novel notion of alternating
bisimulation in Sec.~\ref{altbisim} and were able to prove the
preservation of ATL formulas in bisimilar iCGS
(Theorem~\ref{thm-bisim-mem}).  Further, in order to show that the
common knowledge interpretation enjoys the Hennessy-Milner property,
in Sec.~\ref{bisimgames} we introduced an imperfect information
variant bisimulation games and showed
that the \duplic{} coalition has a winning strategy if and only if
there exists a bisimulation between the two given iCGS
(Theorem~\ref{game}). Finally, in Sec.~\ref{HMproperty} we proved the
Gale-Stewart determinacy Theorem~\ref{GS-theorem}, which allows us to
prove the
Hennessy-Milner Theorem~\ref{HM-theorem}.  We also provided
counterexamples to the Hennessy-Milner property for the objective and
subjective interpretation of ATL.
To conclude, in Sec.~\ref{bisim_undec}
we showed that
checking the existence of an alternating bisimulation between two iCGS is undecidable in general (Theorem~\ref{undec}).

We note that our Hennessy-Milner theorem utilizes the "yesterday" modality for 
technical reasons. As noted in the proof of Theorem \ref{HM-theorem},
Formula \ref{formula-xatl} might be encoded with an ATL$^*$ formula which does not utilize $\textbf{Y}$. 
The translation of this theorem to the full ATL$^*$ is left for future research.

As another direction for future research, we plan to investigate 
under which conditions our Gale-Stewart-type theorem can be generalized to a full determinacy 
theorem for multi-agent games.

\newcommand{\hoek}[1]{}


\end{document}